\documentclass[a4paper,11pt]{article}%

\usepackage{amsfonts}
\usepackage{amsmath}
\usepackage{amssymb}
\usepackage{amsthm}
\usepackage{authblk}
\usepackage{bbm}
\usepackage{cite}
\usepackage[dvipsnames, usenames]{color}
\usepackage{commath}
\usepackage{enumerate}
\usepackage{epsfig}
\usepackage[latin1]{inputenc}
\usepackage{float}
\usepackage{graphicx}
\usepackage{geometry}
\usepackage{leftidx}
\usepackage{mathrsfs}
\usepackage{mathtools}
\usepackage{titlesec}
\usepackage{undertilde}
\usepackage{verbatim}
\usepackage[usenames,dvipsnames,svgnames,table]{xcolor}
\usepackage[all]{xy}

\usepackage{hyperref} 
\allowdisplaybreaks

\DeclareMathOperator*{\argmax}{arg\,max}

\DeclareMathOperator{\sgn}{sgn}
\newcommand{\defeq}{\vcentcolon=}
\newcommand{\backdefeq}{=\vcentcolon}
\newcommand{\Lpm}{\Lambda^{\pm}}
\newcommand{\Lp}{\Lambda^{+}}
\newcommand{\Lm}{\Lambda^{-}}
\newcommand{\lpm}{\lambda^{\pm}}
\newcommand{\lp}{\lambda^{+}}
\newcommand{\lm}{\lambda^{-}}
\newcommand{\delpm}{\delta^{\pm}}
\newcommand{\delp}{\delta^{+}}
\newcommand{\delm}{\delta^{-}}

\newcommand{\Npm}{N^{\pm}}
\newcommand{\Nm}{N^{-}}
\newcommand{\Np}{N^{+}}
\newcommand{\Hp}{H^{+}}
\newcommand{\Hm}{H^{-}}
\newcommand{\Hpm}{H^{\pm}}

\newcommand{\hpm}{h^{\pm}}

\newcommand{\apm}{a^{\pm}}

\newcommand{\bpm}{b^{\pm}}

\newcommand{\cpm}{c^{\pm}}

\newcommand{\pii}{\Pi^i}
\newcommand{\pij}{\Pi^j}
\newcommand{\piia}{\Pi^{\alpha,i}}
\newcommand{\pija}{\Pi^{\alpha,j}}

\newcommand{\pia}{\Pi^\alpha}
\newcommand{\pit}{\Pi^{\alpha,t,n,\pi}}
\newcommand{\Yt}{Y^{t,i}}

\newcommand{\hlpmt}{\widehat{\lpm}^{\alpha,t,n,\pi}}

\newcommand{\hlptu}{\widehat{\lp_u}^{\alpha,t,n,\pi}}
\newcommand{\hlmtu}{\widehat{\lm_u}^{\alpha,t,n,\pi}}

\newcommand{\Xdel}{X^{\delm,\delp}}

\newcommand{\Spm}{S^{\pm}}
\newcommand{\Sp}{S^{+}}
\newcommand{\Sm}{S^{-}}
\newcommand{\fipm}{f^{\pm}_i}
\newcommand{\Lipm}{\Lambda_i^{\pm}}
\newcommand{\Lip}{\Lambda_i^{+}}
\newcommand{\Lim}{\Lambda_i^{-}}
\newcommand{\Ljpm}{\Lambda_j^{\pm}}
\newcommand{\Ljp}{\Lambda_j^{+}}
\newcommand{\Ljm}{\Lambda_j^{-}}
\newcommand{\lmax}{\lambda^{*}}

\newcommand{\delmax}{\delta^{*}}
\newcommand{\delmin}{\delta_{*}}
\newcommand{\Nmax}{N^*}

\newcommand{\mPa}{\mathbb P^{\alpha}}
\newcommand{\mPt}{\mathbb P^{\alpha,t,n,\pi}}
\newcommand{\mPtfull}{\mathbb P^{\alpha,t,n,i}}
\newcommand{\mQa}{\mathbb Q^{\alpha}}
\newcommand{\mEdel}{\mathbb E^{\delm,\delp}}
\newcommand{\mEa}{\mathbb E^{\alpha}}
\newcommand{\mEt}{\mathbb E^{\alpha,t,n,\pi}}

\newcommand{\Dpm}{d^{\pm}}
\newcommand{\Dp}{d^{+}}
\newcommand{\Dm}{d^{-}}
\newcommand{\ve}{\varepsilon}

\newcommand{\Xt}{X^{\alpha,t,x,s}}
\newcommand{\St}{S^{t,s}}
\newcommand{\Nt}{N^{t,n}}

\newtheorem{assumptions}{Assumptions}[subsection]

\newtheorem{theorem}[assumptions]{Theorem}
\newtheorem{lemma}[assumptions]{Lemma}
\newtheorem{corollary}[assumptions]{Corollary}
\newtheorem{proposition}[assumptions]{Proposition}

\theoremstyle{definition}

\newtheorem{examples}[assumptions]{Examples}

\newtheorem*{convention*}{Convention}

\theoremstyle{remark}
\newtheorem{remark}[assumptions]{Remark}
\newtheorem*{remark*}{Remark}

\newtheorem*{notation*}{Notation}

\numberwithin{equation}{subsection}
\title{Optimal market making under partial information with general intensities}
\author{Luciano Campi}
\author{Diego Zabaljauregui\footnote{Corresponding Author. E-mail: {\tt d.zabaljauregui@lse.ac.uk}\\ The authors would like to thank Álvaro Cartea and Katia Colaneri for helpful discussions on the subject.}}
\affil{Department of Statistics\\London School of Economics and Political Science} 

\date{ }

\begin{document}
\maketitle

\begin{abstract}
\noindent Starting from the Avellaneda--Stoikov framework \cite{AS}, we consider a market maker who wants to optimally set bid/ask quotes over a finite time horizon, to maximize her expected utility. The intensities of the orders she receives depend not only on the spreads she quotes, but also on unobservable factors modelled by a hidden Markov chain. We tackle this stochastic control problem under partial information with a model that unifies and generalizes many existing ones under full information, combining several risk metrics and constraints, and using general decreasing intensity functionals. We use stochastic filtering, control and piecewise-deterministic Markov processes theory, to reduce the dimensionality of the problem and characterize the reduced value function as the unique continuous viscosity solution of its dynamic programming equation. We then solve the analogous full information problem and compare the results numerically through a concrete example. We show that the optimal full information spreads are biased when the exact market regime is unknown, and the market maker needs to adjust for additional regime uncertainty in terms of P\&L sensitivity and observed order flow volatility. This effect becomes higher, the longer the waiting time in between orders.
\end{abstract}


\section*{Introduction}
\addcontentsline{toc}{section}{Introduction}
Given a financial market, a market maker (MM) can be understood as someone who provides liquidity for a certain asset. That is, she (almost) continuously posts bid/ask quotes for the asset, in the hope to profit from the bid-ask spread. In choosing how to do so, the MM faces a complicated problem on several levels; namely: the instantaneous margin/volume trade-off (the further away she quotes from the ``fair price", the less she gets executed, and vice-versa), adverse price movements and inventory risk (exposure), execution costs, and many others. 

An increasingly popular mathematical approach to the MM problem, both in academia and in practice, is by means of stochastic optimal control. In particular, a line of research has focused on the modelling framework proposed by Avellaneda and Stoikov \cite{AS} (rooted in turn in \cite{HS}). Although widely motivated in the literature by order-driven markets such as equity markets, the shape of the limit order book is not explicitly taken into account. Thus, the framework is more easily understood and applied to over-the-counter (OTC) quote-driven markets such as the foreign exchange (FX) market, and we therefore choose to present it in this setting.

In this framework, the MM gives firm bid-ask quotes during a finite time interval by choosing bid/ask spreads with respect to a certain \textit{reference price}.\footnote{Also referred to by some authors as \textit{micro-price} \cite{CJP} or \textit{efficient price} \cite{DRR} in the martingale case.} Depending on the market, this price could be for example an aggregated mid-price or a dealer-to-dealer price, and it is frequently assumed to behave as an arithmetic Brownian motion. To explicitly model the margin/volume trade-off, the probability that the MM gets executed decays as a function of the corresponding spread. More precisely, it is assumed that the MM receives market orders according to counting processes of stochastic intensity. Most typically in the mathematical literature, intensities are assumed to decay exponentially on the spreads (mainly for tractability reasons). The goal of the MM is to find a strategy that allows her to maximize her expected terminal utility, which is taken as a constant absolute risk aversion (CARA) utility. The problem is then translated into a deterministic one: solving the associated Hamilton--Jacobi--Bellman (HJB) equation, which is a partial-integro differential equation (PIDE) for the MM's value function, and ultimately retrieving the optimal strategy in feedback form. 

Within the described framework, a lot variants have been put forward and extensively studied. In \cite{BL, GLFT0}, Bayraktar and Ludkovski, and Guéant, Lehalle and Fernandez-Tapia, apply a one-trading-side version to optimal liquidation in the risk-neutral and risk averse contexts, respectively. The introduction of a constraint on the inventory in \cite{GLFT}, allowed the authors to rigorously solve the original problem of \cite{AS} with exponential intensities, by means of a verification theorem. This constraint has been widely used moving forward, with the exception of \cite{FL1, FL2}, where strategies are derived without a verification theorem.

Guéant and Lehalle continued to actively contribute to the area. In \cite{GL}, they revisit the risk averse optimal liquidation problem for general intensities satisfying a certain ordinary differential inequality, and in \cite{G} the same is done for market making.\footnote{The latter paper also considers the multi-asset case.} This assumption had been firstly introduced in \cite[Sect.5.3]{BL} under risk-neutrality.

Other prolific contributors in the field are Cartea, Jaimungal and their coauthors \cite{CJ, CJP, CJR}, who introduced a quadratic running penalty on the inventory to manage the ``accumulated" inventory risk for an otherwise risk-neutral MM. Constraints on the MM's spreads have also been considered in some of them. It is worth noting that \cite{G} shows how the two seemingly different subclasses of models (risk averse and ``risk-neutral" with running penalty) can ultimately be characterized by a unique system of ordinary differential equations (ODEs) when considering two appropriate ansatz. 

A relevant issue in practice that has not been included in the previous models is the following one. Empirical evidence (see, e.g., \cite{CJ0}) suggests that liquidity taken by clients depends not only on the quoted spreads but also on other unobservable factors. Indeed, the confluence of factors such as market sentiment towards the asset and the competition with other market makers also affects the intensities at which the MM receives orders. This complicated effect can be modelled in a simplified fashion by making the intensities depend as well on a hidden finite-state Markov chain, effectively reflecting the regime or state of the market (different levels from very slow to very active). To the best of our knowledge, this has only been briefly done in the Avellaneda--Stoikov framework by proposing an approximation for the optimal strategy with exponential intensities \cite[Sect.5.1]{CJ0} or studying a simple two-states version with power-law intensities \cite[Sect.5.4]{BL}. Both of these papers deal only with the ``risk-neutral" (possibly penalized) case and make the unrealistic assumption that the current market regime is known by the MM.\footnote{\cite[Sect.3.3]{BL0} (arXiv version) and \cite{CJ19} also study optimal trade execution under partial information, albeit with uncontrolled intensities.}

When the state of the market is unknown, the problem becomes significantly more challenging for several reasons:
\begin{itemize}
\item It becomes a combined problem of stochastic control and filtering. The MM needs to dynamically make her best possible prediction of the market regime (or more precisely, its distribution), based on the information she has (i.e., the orders she has received so far and the evolution of the reference price), and adjust her spreads accordingly. This prediction is known as a filter.
\item The associated HJB PIDE has higher dimension and more non-linearities.
\item The standard approach used in all of the previously cited papers relies on reducing the HJB PIDE to a system of ODEs by means of an ansatz, proving that such a system has a classical (smooth) solution, and recovering the value function via a verification theorem. Under partial information however, the reduced HJB equation is still a complicated PIDE such that, in general, a classical solution may not exist (or it may be too difficult to prove otherwise). Hence, the ansatz argument breaks down. 
\item As a consequence, one needs to resort to the concept of viscosity solutions \cite{FS}. In addition, the numerical resolution unfailingly becomes a lot more involved than for simple systems of ODEs.
\item On a technical level, the construction of the model is not straightforward. The MM needs to adjust her strategy based on her observable information, such as the arriving orders, but this flow of information is in turn affected by the MM's actions.
\end{itemize}

In this paper, we solve the problem of the MM under partial information. First, we start by unifying in one single formulation all the modelling features described so far. This allows us to simultaneously tackle all the models at once with a single approach, while generalising them at the same time. Indeed, our formulation allows for the interaction of any CARA utility (whether risk-neutral or risk averse) with running inventory penalty, terminal execution cost, inventory constraints and spread constraints. Further, motivated by practitioners' needs, we strongly generalize the intensity shapes to any continuous, decreasing to zero functions, adding modelling flexibility.\footnote{This is done at the expense of renouncing to uniqueness in the optimal strategy. We also assume the decay to be ``fast enough" in certain cases.} We even allow for the inventory to be unconstrained when no penalties are present. (This scenario is considered mainly for the sake of completeness and comparison, at almost no extra cost.)

Secondly, we let the intensities depend on a $k$-dimensional hidden Markov chain. Following \cite{CEFS}, we use a weak formulation\footnote{That is, with controlled probability measures.} to construct a well-defined model (with exogenous information) and solve the filtering problem by means of the reference probability approach of stochastic filtering \cite[Chpt.VI]{B}. The rigorous setting of the model and its full characterization are carried out in Section \ref{s:model}, while the filtering problem is solved in Section \ref{s:filtering}.

The optimization problem of the MM is then reformulated in terms of the usual state variables together with the $k$-dimensional observable distribution of the Markov chain (Section \ref{s:viscosity}). At this point, the problem is too involved both analytically and numerically. We proceed by showing that, formally, there is an ansatz for the value function that reduces the dimensionality of the problem. However, as the verification approach is not valid any more, we explicitly find the guessed decomposition in terms of the value function of a new, diffusion-free, problem (Theorem \ref{main_theorem_1}). We then characterize the reduced value function as the unique continuous viscosity solution of its formally derived equation (Theorem \ref{main_theorem_2}). The latter is done by harnessing results of piecewise-deterministic Markov process (PDMPs) \cite{CEFS, DF}, as first defined in \cite{D1}.

Thirdly, we solve the idealized problem of a MM with full information, who can observe the Markov chain (Section \ref{s:full_info}). We show that a similar ansatz and the standard approach work out in this case. We prove that the MM's value function is a classical solution of its HJB equation via a general verification theorem (Theorem \ref{verif_full_info}) and we recover the well known strategies for one regime as particular cases. 

Finally, we compare the optimal strategies under full and partial information through a concrete example, by numerical analysis (Section \ref{s:numerics}). In particular, we show that the optimal full information spreads are biased when the exact regime is unknown, and using them becomes suboptimal. We interpret the adjustment needed in terms of observable order flow volatility and sensitivity of the expected profit to observable regime changes; and we show how this effect becomes higher, the longer the waiting time in between orders (leading to higher uncertainty for the MM).


\section{Setting and main assumptions}
\label{s:model}
\setcounter{assumptions}{0}
\setcounter{subsection}{0}

\subsection{Preliminaries on the probability spaces}
\label{s:preliminaries}

We start by setting up our framework in an abstract fashion, deferring the question of existence of a model to Proposition \ref{Girsanov1}, and a characterization of all such models to Proposition \ref{Girsanov2}. As mentioned in the Introduction, we seek to construct a model under a weak formulation, so that the information flow remains exogenous (i.e., unaffected by the MM's actions).

Let $T>0$ be a given finite horizon and $(\Omega,\mathcal F,\mathbb F=(\mathcal F_t)_{0\leq t\leq T})$ a right-continuous filtered measurable space with $\mathcal F_T=\mathcal F$. Suppose it supports three adapted stochastic processes $W,\Np,\Nm$. (More assumptions to be added in the sequel.) Let $\mathbb F^{W,\Np,\Nm}=(\sigma(W_u,\Np_u,\Nm_u:0\leq u\leq t)\vee\mathcal F_0)_{0\leq t\leq T}$ be the natural filtration of these processes enlarged (``completed") by $\mathcal F_0$ and define the set of \textit{admissible spreads} as
\begin{equation}
\label{A}
\mathcal A\defeq\{\delta:[0,T]\times\Omega\to\overline{(\delmin,\delmax)}: \delta\mbox{ is }\mathbb F^{W,\Np,\Nm}\mbox{-predictable and bounded}\},
\end{equation}
where $-\infty\leq\delmin<\delmax\leq+\infty$ are fixed constants and $\overline{(a,b)}=\mathbb R\cap [a,b]$ denotes the closure of the interval $(a,b)$ in $\mathbb R$, for any $-\infty\leq a<b\leq+\infty$. Note that for $\delmin=-\infty$ and $\delmax=+\infty$, the admissible spreads are not uniformly bounded. The self-imposed constraints $\delmin,\delmax$ for the MM's spreads can be taken to be different for the bid and the ask if wanted, without any additional effort.

Consider on the former space a family $(\mPa)_{\alpha\in{\mathcal A}^2}$, $\alpha=(\delm,\delp)$, of equivalent probability measures such that the sigma algebra generated by their null sets is $\mathcal F_0$. Note that for each $\alpha$, $(\Omega, \mathcal F, \mathbb F, \mPa)$ is under the usual conditions and $\mathcal F_0$ is the completed trivial sigma algebra. We refer to $\mPa$ as the \textit{physical} or \textit{historical} probability given the \textit{admissible strategy} (or \textit{control}) $\alpha$, and we write $\mEa$ for the expectation under $\mPa$. We will drop the ``$\alpha$" from the notation when there is no room for ambiguity. 

Henceforth, all the processes (resp. properties) considered are supposed to be defined (resp. hold) on the space $(\Omega, \mathcal F, \mathbb F, \mPa)$ for all $\alpha\in{\mathcal A}^2$, unless otherwise stated. For example, a ``Brownian motion" is a process on $[0,T]\times\Omega$ that is an $(\mathbb F,\mPa)$-Brownian motion for all $\alpha$. All subfiltrations are understood to have been augmented to satisfy the usual conditions (which in the case of the natural filtrations of Feller processes, it amounts simply to completing them). Càdlàg versions of the processes are used whenever available. 

\subsection{Description of the model}

A market maker in a quote-driven market gives binding bid/ask quotes $\Sm, \Sp$ resp. for a certain financial asset, during the time interval $[0,T]$. She receives unitary size market orders\footnote{Any constant size would work in the same way, but this convention simplifies the notation. See, e.g., \cite{G} for some formulas with arbitrary constant size under full information.} to buy/sell according to counting processes $\Nm, \Np $ starting at 0, with a.s. no common jumps and stochastic intensities $\lm, \lp$ resp. (see \cite[p.27, Def. D7]{B}). Each intensity at time $t$ depends on the corresponding spread $\delpm_t\defeq\pm\big(\Spm_t - S_t\big)$ between the MM's quotes and a reference price $S$. $S$ may be interpreted, e.g., as an aggregated mid-price or a dealer-to-dealer price, depending on the market. We assume $S$ is a Brownian motion with drift, of the form 
\begin{equation}
\label{S}
S_t= s_0 + \mu t + \sigma W_t,
\end{equation}
where $W$ is a Wiener process, $s_0,\mu\in\mathbb R$ and $\sigma\geq 0$. Note that the MM can fully specify her quote $\Spm_t$ by choosing the spread $\delpm_t$, since $\Spm_t = S_t \pm\delpm_t$. We assume $\delm,\delp\in\mathcal A$. In particular, the MM chooses her spreads based on the observation of $W,\ \Nm\mbox{ and }\Np$. We refer to $\mathbb F^{\Nm,\Np,W}$ as the \textit{observable filtration}.

In addition, the intensities also depend on a hidden Markov chain $Y$ with state space $\{1,\dots,k\}$, initial distribution $\mu_0$ and deterministic generator matrix $Q=(q^{ij}_t)_{1\leq i,j\leq k}$. We assume $Q$ is a continuous, stable and conservative $Q(t)$-matrix, i.e.: 
\begin{assumptions}[\textbf{Generator matrix}]
\label{assumptions_Q}
For all $t\in [0,T]$ and $1\leq i\leq k$:

$0\leq q^{ij}_t<+\infty\mbox{ for all }\ 1\leq j\leq k,\ j\neq i,\qquad \displaystyle\sum_{j=1}^k q^{ij}_t =0\quad\mbox{and}\quad Q\in C[0,T].$
\end{assumptions}
\noindent The previous are standing assumptions in the literature of Markov chains and they guarantee in particular the existence of a unique (up to evanescence) càdlàg version of $Y$, and the existence of the predictable intensity kernel for its jump measure, as given in (\ref{intensity_kernels}). In addition, we suppose that 
\begin{equation}
\label{no_common_jumps}
Y\mbox{ has a.s. no common jumps with }\Nm\mbox{ and }\Np. 
\end{equation}
In this model, $Y_t$ represents the \textit{market regime} at time $t$ and results from the interaction of a range of different factors, such as market sentiment towards the asset and varying levels of competition with other liquidity providers. These effects are almost never explicitly modelled in the mathematical literature of market making. To this end, it is natural to assume that $Y$ is not directly observable by the MM who can only see $\Nm,\Np$ and $W$. (As often assumed in filtering and control with hidden Markov chain models, the parameters $k,Q,\mbox{ and }\mu_0$ of $Y$ are known to the MM, who needs to model/estimate them in practice.)

We define the MM's inventory process 
\begin{equation}
\label{N}	
N_t\defeq n_0 + \Nm_t - \Np_t,
\end{equation}
and the cash account process
\begin{equation}
\label{X}
\Xdel_t\defeq x_0 + \int_0^t \left(S_u+\delp_u\right)d\Np_u - \int_0^t \left(S_u-\delm_u\right)d\Nm_u,
\end{equation}
for some fixed initial values $n_0\in\mathbb Z$ and $x_0\in \mathbb R$. The MM's preferences at terminal time are represented by a CARA utility function:
\begin{equation*}
U_\gamma(c)=\frac{1-e^{-\gamma c}}{\gamma},\mbox{ for }c\in\mathbb R\mbox{ and }\gamma>0,\mbox{ and }U_0=Id.
\end{equation*}
\noindent Note that 
$$\mathbb F^{\Nm,\Np} = \mathbb F^N\quad\mbox{ and }\quad\mathbb F^{\Nm,\Np,W}=\mathbb F^{N,W},$$ 
since $\Nm,\Np$ have a.s. no common jumps. 

We shall make some very natural modelling assumptions on the intensities. These generalize those in \cite{BL, G, GL} to a context with several market regimes, without the need of any conditions on the derivatives or even smoothness.

\begin{assumptions}[\textbf{Orders intensities}]
\label{assumptions_intensities}
There exist functions $\Lipm:\overline{(\delmin,\delmax)}\to (0,+\infty)$ for $i=1,\dots,k$ and $\Nmax\in \mathbb N \cup\{+\infty\}$ with $-\Nmax\leq n_0\leq\Nmax$, such that:
\begin{enumerate}[(i)] 
\item \label{assumptions_intensities_lambda}
$\lpm_t = \lambda_t^{\pm,\alpha} = \Lpm_{Y_{t^-}}(\delpm_t)\mathbbm 1_{\{\mp N_{t^-}<\Nmax\}}$ if $\alpha=(\delm,\delp)$. 
\item \label{assumptions_intensities_Lambda}
$\Lipm$ is continuous, decreasing (not necessarily strictly) and $\displaystyle\lim_{\delta\to+\infty}\delta^{\mathbbm 1_{\{\gamma=0\}}}\Lipm(\delta)=0$ if $\delmax=+\infty$, for all $1\leq i\leq k$.
\end{enumerate}
\end{assumptions}

\begin{notation*}
For the sake of readability, we will often write $\mathbbm 1^{\pm}_t$ instead of $\mathbbm 1_{\{\mp N_{t^-}<\Nmax\}}$.
\end{notation*}

\noindent $\Nmax$ is a self-imposed constraint on the MM's inventory (whichever its sign) as originally proposed in \cite{GLFT}. When $\Nmax=\infty$ we are in the unconstrained context, while a finite $\Nmax$ effectively means that the MM will not buy (resp. sell) whenever $N_{t^-}=\Nmax$ (resp. $N_{t^-}=-\Nmax$). In practice, the MM could achieve this either by abstaining from quoting or by quoting with an excessively large spread that prevents any transaction (i.e., a ``stub" or ``placeholder quote").\footnote{For simplicity and without loss of generality, only the first alternative is formalized in our model (as done in \cite{G, GL, GLFT}) and this is reflected in the admissible spreads being real-valued. We could also allow for different constraints depending on the sign of the inventory, but we refrain from this to simplify the notation.} 

\begin{remark}
\label{expected_margin}
The last assumption states in particular that the intensities should decrease to zero when the spreads grow arbitrarily large. Furthermore, when $\gamma=0$ we require that they decay faster than $1/\delta$. Loosely speaking, this states that for any $\gamma\geq 0$ the ``expected" utility of the MM's instantaneous margin, $\digamma^\pm_i(\delta)\defeq\Lambda_i(\delta) U_\gamma(\delta)$, should vanish for ``stub quotes". This is a reasonable assumption in practice, as the opposite could lead to unrealistic optimal strategies such as continuously quoting ``infinite spreads".
\end{remark}

\begin{remark*}
Note that the above intensities are predictable, and since any admissible spreads $\delm,\delp$ are bounded, $\lm,\lp$ are in turn bounded. Accordingly, $\Nm,\Np, N$ and $X$ are non-explosive (see, e.g., \cite[p.27 T8]{B}). Furthermore, for any constant $\lmax>0$ such that $\lm + \lp\leq\lmax$, it is easy to see that for all $p\in\mathbb N,\ t\in[0,T]$ and $r\geq 0$,
\begin{equation}
\label{bounded_moments}
\mEa[|N_t|^p]\leq \mEa[(n_0+\Nm_t+\Np_t)^p]\leq\mEa[(n_0 + M)^p],\quad {\mathcal M}^\alpha_{|N_t|}(r)\leq {\mathcal M}^\alpha_{n_0+\Nm_t+\Np_t}(r) \leq {\mathcal M}^\alpha_{n_0+M}(r),
\end{equation}
where $M\sim\mbox{Poisson}(\lmax T)$ and ${\mathcal M}^\alpha_R$ is the moment generating function of a random variable $R$.\footnote{By possibly enlarging the space, one can consider a counting process $Z$ with no common jumps with $\Nm,\Np$ and stochastic intensity $\lmax-\lm-\lp\geq 0$. Then the process $Z+\Nm+\Np$ is a Poisson with intensity $\lmax$ that dominates $\Np+\Nm$. The claim follows immediately.}
\end{remark*}

The following are some examples which we shall come to back in Section \ref{s:computing_spreads}, when revisiting the standard assumptions in the literature. We remark that exponential, power-law and logistic intensities are the explicit ones most commonly used in the mathematical literature, for tractability reasons.
\begin{examples}
\label{ex}
\begin{enumerate}
\item \label{ex1}
$\Lipm(\delta)=\apm_i e^{-\bpm_i \delta}$, with $\apm_i,\bpm_i>0,\ -\infty\leq\delmin<\delmax\leq+\infty$.
\item \label{ex2}
$\Lipm(\delta)=\frac{\apm_i}{1+\cpm_i e^{\bpm_i \delta}}$, with $\apm_i,\bpm_i,\cpm_i>0,\ -\infty\leq\delmin<\delmax\leq+\infty$.
\item \label{ex3}
$\Lipm(\delta)=\apm_i \big(\pi/2 + \arctan(-\bpm_i\delta+\cpm_i)\big)$, with $\apm_i,\bpm_i>0,\ \cpm_i\in\mathbb R,\ -\infty\leq\delmin<\delmax\leq+\infty$. 
\item \label{ex4}
$\Lipm(\delta)=\apm_i (\bpm_i\delta+\cpm_i)^{-\Dpm_i}$, with $\apm_i,\bpm_i>0,\ \Dpm_i>\mathbbm 1_{\{\gamma=0\}},\ \cpm_i\in\mathbb R,\ \max_i\{-\cpm_i/\bpm_i\}<\delmin<\delmax\leq+\infty$.
\end{enumerate}
\end{examples}
To fix ideas, consider Example \ref{ex1} and let us observe that if $Y$ represents decreasing levels of competition for the MM, the values of $\apm_i$ (resp. $\bpm_i$) should increase (resp. decrease) as $Y$ increases. The same is true if $Y$ represents increasing levels of positive sentiment (or ``bullishness") towards the asset. In practice, $Y$ will result from the combination of these effects and many more.

Suppose the MM can hedge any remaining inventory at time $T$ at the reference price $S_T$ minus a certain execution cost. Neglecting discounting between $0$ and $T$, we consider the optimization problem faced by the MM who tries to maximize the expected utility of her terminal \textit{penalized profit and loss} (P\&L),

\begin{equation}
\label{problem}
\sup_{\delm,\delp\in\mathcal A}\mEdel\left[U_\gamma\Big(\Xdel_T+S_T N_T - \ell(N_T)-\frac{1}{2}\sigma^2\zeta\int_0^TN_t^2dt \Big)\right],
\end{equation}
where: 
\begin{enumerate}[(i)]
\item $\zeta\geq 0$. When $\zeta>0$ we are including in the model a \textit{running penalty}, as firstly done in \cite{CJ}, for the MM to further control her \textit{accumulated inventory risk}. This is the same as subtracting the variance of the mark-to-market value of the inventory, weighted by $\zeta$.
\item $\ell:\mathbb R\to\mathbb [0,+\infty)$ represents the final execution cost and is increasing on $[0,+\infty)$, decreasing on $(-\infty,0]$ and $\ell(0)=0$. (And usually convex in practice.)
\item If $\Nmax=+\infty$, we set $\gamma=\zeta=0\equiv\ell$.
\end{enumerate}

The last restriction on the parameters states that the only case we will consider in which the inventory could be arbitrary large (whichever its sign) is that of a completely risk-neutral MM with negligible costs. The risk averse case is far more challenging and has not been treated in full mathematical detail even in the complete information case (see discussion in \cite{GLFT}). The model given by (iii) is of secondary interest in practice, yet, it will allow us to further understand the general problem and methodology and to have a more holistic view at almost no extra cost.

\begin{remark*}
Previous works in optimal market making do not consider both penalty and CARA utility (with $\gamma>0$) in unison, and instead treat the two families of models separately \cite{G}. Besides the obvious interest in unifying the approach and generalizing existing models, allowing for $\gamma>0$ and $\zeta>0$ simultaneously adds flexibility to the risk managing capabilities of the MM. Indeed, in \cite{G} the author derives some HJB-type systems of ODEs for each problem with a unique risk aversion parameter $\bar\gamma$, and then relates them by the introduction of an auxiliary parameter $0\leq \bar\xi\leq \bar\gamma$. The later is afterwards interpreted as measuring risk aversion to non-execution risk only and it is carried forward to the formulae of the ``optimal strategies". However, such strategies are not shown to result from any particular optimization problem for $0<\bar\xi<\bar\gamma$. By looking at the equations of the full information version of our model (Section \ref{s:full_info}) with a single market regime, one can immediately see that the equations in \cite{G} are obtained via the reparametrization $\gamma=\bar\xi$ and $\zeta = \bar\gamma - \bar\xi$. Thus, the formulation in (\ref{problem}) allows us, in particular, to establish the optimality of the strategies in \cite{G} for any value of the non-execution risk aversion and, in general, to differentiate between aversion to different types of risks. We can say that, in our model, $\gamma$ represents aversion to all types of risks, while the penalty $\zeta$ is used to further increase the aversion to risks other than non-execution risk.
\end{remark*}

\subsection{Jump measures}

It will be useful in the sequel to understand which are the jump measures and intensity kernels involved in our model. Let $\mu^N(dt,dz)$ (or simply $\mu^N$) be the jump (or counting) measure of $N$ (see def. in \cite[p.234]{B} or \cite[p.69]{JS}). We use a similar notation for the jump measures of $Y$ and the pair $(Y,N)$. These are all finite random measures since the processes are non-explosive for any $\delp,\delm\in\mathcal A$ (see (\ref{bounded_moments})), and they admit predictable intensity kernels (see \cite[p.235 D2]{B}). If $m_{z_0}(dz)$ is the Dirac measure at a point $z_0\in \mathbb R$, then the $(\mPa,\mathbb F)$-predictable intensity kernels of $\mu^N,\mu^Y$ and $\mu^{(Y,N)}$ are respectively given by:

\begin{equation}
\label{intensity_kernels}
\begin{split}
\eta^{\alpha,N}_t(dz) & =\eta^N(Y_{t^-}, \delm_t,\delp_t,N_{t^-},dz)\\
																		& = \lm_t m_{1}(dz) + \lp_t m_{-1}(dz)\\
																		& = \Lm_{Y_{t^-}}(\delm_t) \mathbbm 1_{\{N_{t^-}<\Nmax\}} m_{1}(dz)+ \Lp_{Y_{t^-}}(\delp_t)\mathbbm 1_{\{-N_{t^-}<\Nmax\}} m_{-1}(dz),\\
				\eta^Y_t(dh)  & = \eta^Y(Y_{t^-},dh)=\sum_{j\neq Y_{t^-}}q^{Y_{t^-},j}_t m_{j-Y_{t^-}}(dh)												
\end{split}
\end{equation}
and 
\begin{equation}
\label{intensity_kernel_pair}
\eta^{\alpha,(Y,N)}_t(dh,dz) =\eta^{\alpha,N}_t(dz)\otimes m_0(dh) + \eta^Y_t(dh)\otimes m_0(dz).
\end{equation}

\noindent We denote by ${\bar\mu}^N_\alpha(dt,dz)$, $\bar{\mu}^Y_\alpha(dt,dh)$ and ${\bar\mu}^{(Y,N)}_\alpha(dt,dh,dz)$ the corresponding $(\mathbb F,\mathbb P^\alpha)$-compensated measures.
\begin{remark*}
Note that equation (\ref{intensity_kernel_pair}) is a consequence of assuming that $Y,N$ have a.s. no common jumps (which in turn is equivalent to (\ref{no_common_jumps})). Indeed, it is easy to see that $Y,N$ have a.s. no common jumps if and only if
\begin{equation}
\mu^{(Y,N)}(dt,dh,dz)= \mu^{N}(dt,dz)\otimes m_0(dh) + \mu^Y(dt,dh)\otimes m_0(dz),
\end{equation}
which implies (\ref{intensity_kernel_pair}).
\end{remark*}

\subsection{Construction and characterization of the model}

In this section, we rigorously construct and characterize our model. On top of the mathematical interest in a well defined framework, the results and techniques presented here will be used throughout the paper, when tackling both the filtering and the optimization problems. We start by proving the existence of a model under our assumptions (Proposition \ref{Girsanov1}). We construct it by means of a reference probability (an overview of which can be found, e.g., in \cite[Chpt.VI]{B}) and, in particular, the Girsanov theorem for point processes. 

For the following proposition, let us consider a filtered probability space $(\Omega,\mathcal F,\mathbb F,\mathbb Q)$ under the usual conditions, with $\mathcal F=\mathcal F_T$ and $\mathcal F_0$ the completed trivial sigma algebra. We refer to $\mathbb Q$ as the \textit{reference probability}. Suppose this space supports a two-dimensional counting process $(\Nm,\Np)$ (see def. in \cite[Chapt.II]{B}), a Wiener process $W$ and a Markov chain $Y$ with finite state space $\{1,\dots,k\}$, initial distribution $\mu_0$ and time-dependent generator matrix $Q=(q^{ij}_t)_{1\leq i,j\leq k}$ satisfying Assumptions \ref{assumptions_Q}, such that $\Nm,\Np,Y$ verify (\ref{no_common_jumps}). We set $N=\Nm-\Np$ as before and suppose $\Npm$ has intensity $\mathbbm 1^{\pm}_t=\mathbbm 1_{\{\mp N_{t^-}<\Nmax\}}$, where $\Nmax\in\mathbb N\cup \{+\infty\}$ is given.\footnote{Such a simpler model can be constructed for example as a product of canonical spaces, with the existence of the counting processes with the right intensities proved in \cite[Thm.24 and Cor.31]{JP}. The finite-dimensional result has the same proof as in one dimension, starting from independent Poisson measures.} Let $\Lipm:\overline{(\delmin,\delmax)}\to\mathbb R$ for $i=1,\dots,k$ be functions under the Assumptions \ref{assumptions_intensities} \textit{(\ref{assumptions_intensities_Lambda})},for given $-\infty\leq\delmin<\delmax\leq+\infty$ and $\Nmax\in \mathbb N \cup\{+\infty\}$ with $-\Nmax\leq n_0\leq\Nmax$, and define $\mathcal A$ as in (\ref{A}).

\begin{proposition}
\label{Girsanov1}
Let $\alpha=(\delm,\delp)\in{\mathcal A}^2$ and define the process $Z^\alpha$ as the \textit{stochastic exponential}
$$Z^\alpha\defeq \mathcal E\left( \int_0^\cdot\big(\Lm_{Y_{u^-}}(\delm_u)-1\big)\big(d\Nm_u- \mathbbm 1^-_u du\big)+\big(\Lp_{Y_{u^-}}(\delp_u)-1\big)\big(d\Np_u- \mathbbm 1^+_u du\big)\right).\mbox{ Then,}$$ 

\begin{enumerate}[(i)]
\item $Z^\alpha$ is a strictly positive uniformly integrable $(\mathbb F, \mathbb Q)$-martingale. In particular, $\mathbb E^{\mathbb Q}[Z^\alpha_T]=1$.
\item $\mPa$ defined by $\frac{d\mPa}{d\mathbb Q} = Z^\alpha_T$ is an equivalent probability measure such that, for $(\mathbb F,\mPa)$, $\Npm$ is a counting process with intensity $\lpm_t = \Lpm_{Y_{t^-}}(\delpm_t)\mathbbm 1^{\pm}_t$, $W$ is a Wiener process and $Y$ is a Markov chain with state space $\{1,\dots,k\}$, initial distribution $\mu_0$ and generator matrix $Q$. 
(That is, Assumptions \ref{assumptions_Q}, \ref{assumptions_intensities} and (\ref{no_common_jumps}) are all verified.)
\end{enumerate}
\end{proposition}

\begin{proof}
See Appendix \ref{proof:Girsanov1}.
\end{proof}

Reciprocally, suppose we start with a family of filtered probability spaces $\{(\Omega, \mathcal F, \mathbb F = (\mathcal F_t)_{0\leq t\leq T}, \mPa)\}_{\alpha\in{\mathcal A}^2}$ as in Section \ref{s:preliminaries}, supporting a counting process $(\Nm,\Np)$ and a Markov chain $Y$ satisfying condition (\ref{no_common_jumps}), together with a Wiener process $W$. Suppose also that $Y$ has finite state space $\{1,\dots,k\}$, initial distribution $\mu_0$ and time-dependent generator matrix $Q$, and that our Assumptions \ref{assumptions_Q} and \ref{assumptions_intensities} are in place. We set $N= \Nm-\Np$ as before. We would like to characterize any such model in terms of a reference probability as in Proposition \ref{Girsanov1} by an inverse change of measure. However, we only claim the uniqueness of the reference probability on $\mathbb F^N$. We will come back to this result in the sequel.  

\begin{proposition}
\label{Girsanov2}
Let $\alpha=(\delm,\delp)\in{\mathcal A}^2$ and define 
$$\bar{Z}^\alpha\defeq \mathcal E\left(\int_0^\cdot\Big(1/\Lm_{Y_{u^-}}(\delm_u)-1\Big)\big(d\Nm_u-\lm_u du\big)+\Big(1/\Lp_{Y_{u^-}}(\delp_u)-1\Big)\big(d\Np_u-\lp_u du\big)\right).\mbox{ Then,}$$  

\begin{enumerate}[(i)]
\item $\bar{Z}^\alpha$ is a strictly positive uniformly integrable $(\mathbb F, \mPa)$-martingale. In particular, $\mEa[\bar{Z}_T]=1$. 
\item $\mQa$ defined by $\frac{d\mQa}{d\mPa} = \bar{Z}^\alpha_T$ is a probability measure equivalent to $\mPa$ such that, for $(\mathbb F,\mQa)$, $\Npm$ is a counting process with intensity $\mathbbm 1^{\pm}_t$,  
$W$ is a Wiener process and $Y$ is a Markov chain with state space $\{1,\dots,k\}$, initial distribution $\mu_0$ and generator matrix $Q$. Furthermore, $Y, N, W$ are independent.
\item If we define $Z^\alpha$ as in Lemma \ref{Girsanov1} under $\mathbb Q^\alpha$, then $Z^\alpha=1/\bar{Z}^\alpha$ (i.e., the two changes of measure are inverse of each other).
\item For all $\tilde\alpha\in{\mathcal A}^2$, $\mathbb Q^{\tilde\alpha}\equiv\mQa$ on $\mathbb F^N_T$. 
\end{enumerate}
\end{proposition}

\begin{proof}
See Appendix \ref{proof:Girsanov2}.
\end{proof}

\section{Filtering problem}
\label{s:filtering}
\setcounter{assumptions}{0}
\setcounter{subsection}{1}

Since the MM cannot directly observe all the information in $\mathbb F$ but only $\mathbb F^{N,W}$ (in particular, she cannot observe $Y$), in order to solve the optimization problem (\ref{problem}) under partial information we want to reduce it first to an equivalent one under full information. Throughout this section we work under $\mPa$ with $\alpha=(\delm,\delp)\in{\mathcal A}^2$ fixed. We sometimes omit $\alpha$ from the notation for simplicity.

Let us consider the optional projections
$$
\piia \defeq^o(\mathbbm 1_{\{Y_t=i\}})^{(\mPa,\mathbb F^{N,W})},\ 1\leq i\leq k,\quad\mbox{ and }\quad\pia\defeq (\Pi^{\alpha,1},\dots,\Pi^{\alpha,k}).\footnote{For any càdlàg bounded process $M$ (not necessarily adapted) on a filtered probability space $(\Sigma,\mathcal H,\mathbb H,\mathbb P)$ satisfying the usual conditions, the \textit{optional projection} of $M$ on $\mathbb H$ is the unique càdlàg process $^oM^{(\mathbb P,\mathbb H)}$ such that $ ^oM^{(\mathbb P,\mathbb H)}_t = \mathbb E^{\mathbb P}[M_t|\mathcal H_t]$ a.s. for each $t$. Its existence is guaranteed by the Optional Projection Theorem (see, e.g., \cite[p.264]{JYC} or \cite[p.357-358]{N}).}
$$
That is, $\pia$ is the unique càdlàg version of the conditional distribution of $Y$ given the observable information: $\piia_t = \mPa(Y_t=i|\mathcal F^{N,W}_t)$. It reflects the MM's beliefs on the state of the market as she updates them through time.

At any point in time, the MM typically observes orders arriving at rates which do not seem to match exactly those of any particular regime, but some average ones instead. We want to characterize the observable (that is, the $\mathbb F^{N,W}$-) predictable intensities of $\Nm,\Np\mbox{ and }\mu^N$ in terms of $\pia$. Loosely speaking, the observable intensity of, say $\Np$, is obtained projecting: $\mEa[\lp_t|\mathcal F^{W,N}_t]$ \cite[p.32 Comment and Pseudo-Proof of T14]{B}. This yields an average of the different regime intensities weighted by the observable probabilities of the market states. The only technical difficulty is that of finding a predictable version of such process. $ ^o(\lp)$ has the desired projective property but it is not predictable in general. $ ^o(\lp)_{t^-}$, on the other hand, is predictable but does not normally enjoy the projective property. In fact, the process we are looking for is ``in between" these two.

\begin{proposition}[\textbf{Observable intensities}]
\label{observable_intensities}
The $(\mPa,\mathbb F^{N,W})$-predictable intensities of $\Npm$ and $\mu^N$, resp., are
\begin{equation*}
\widehat{\lpm_t}^\alpha\defeq\mathbbm 1^{\pm}_t\sum_{i=1}^k\piia_{u^-}\Lipm(\delpm_u)\quad\mbox{ and }\quad
\widehat{\eta}^{\alpha,N}_t(dz) \defeq \widehat{\lm_t}^\alpha m_{1}(dz) + \widehat{\lp_t}^\alpha m_{-1}(dz).
\end{equation*}
\end{proposition}

\begin{notation*}
We set $\widehat{\Lpm}(\pi,\delta)=\sum_{i=1}^k\pi^i\Lipm(\delta)$; thus, $\widehat{\lpm_t}^\alpha=\mathbbm 1^{\pm}_t\widehat{\Lpm}(\pia_{u^-},\delpm_u)$.
\end{notation*}

\begin{proof}
See Appendix \ref{proof:observable_intensities}.
\end{proof}

We give now the \textit{filtering (or Kushner-Stratonovich) equations} for the observable distribution of $Y$. These are coupled stochastic differential equations (SDEs) governing the dynamics of $\pia$. They allow the MM to dynamically and efficiently update an existing estimate $\pia_t$ upon receiving new information, without having to recompute it from scratch. In a nutshell, given a side of the market and the spread quoted, an order arrival modifies the observed probability of regime $i$ by a multiplicative factor. This factor is given by the percentage difference between the orders intensity of regime $i$ and the observed one. In between orders however, the probabilities change continuously and in a more subtle way, depending not only on the different levels of liquidity but also on the transition rates between them.

\begin{notation*}
We denote by $\Delta\subset\mathbb R^k$ the $(k-1)$-simplex (i.e., $\Delta=\{\pi\in\mathbb R^k: 0\leq\pi^i\leq 1\ \mbox{for all }i\mbox{, and }\sum_{i}\pi^i=1\}$ and by $\Delta^\circ$ its interior relative to the hyperplane $\{\pi\in\mathbb R^k: \sum_{i}\pi^i=1\}$ (i.e., $\Delta^o=\{\pi\in\Delta: 0<\pi^i< 1\mbox{ for all }i\}$).
\end{notation*}

\begin{proposition}[\textbf{Observable distribution of $Y$}]
\label{proposition_Pi}
The process $\pia=(\Pi^{\alpha,1},\dots,\Pi^{\alpha,k})$ is the unique strong solution of the constrained system of SDEs
\begin{equation}
\label{Pi_equivalent}
\begin{split}
d\pii_t &= \sum_{j=1}^k\left(q_t^{ji}\pij_t+\pii_t\pij_t\Big(\mathbbm 1^-_t(\Ljm-\Lim)(\delm_t) + \mathbbm 1^+_t(\Ljp- \Lip)(\delp_t) \Big) \right)dt\\ 
				& + \pii_{t^-}\left(\frac{\Lim(\delm_t)}{\widehat{\Lm}(\pia_{u^-},\delm_u)}-1 \right) d\Nm_t + \pii_{t^-}\left(\frac{\Lip(\delp_t)}{\widehat{\Lp}(\pia_{u^-},\delp_u)}-1 \right) d\Np_t,
\end{split}
\end{equation}
subject to $\Pi_0 = \mu_0$ and $\Pi_t\in\Delta$ for all $t\in[0,T]$ a.s.
\end{proposition}

\begin{proof}
See Appendix \ref{proof:proposition_Pi}.
\end{proof}

\begin{remark}
\label{identification}
Consider the identification $\Delta\simeq [0,1]^{k-1}$ (resp. $\Delta^\circ\simeq (0,1)^{k-1}$) obtained by the substitution $\pi^k=1-\sum_{j<k}\pi^j$ (where the choice of the $k$-th coordinate over the rest is completely arbitrary). Then the constrained system of SDEs (\ref{Pi_equivalent})  
for $(\Pi^{\alpha,1}\dots,\Pi^{\alpha,k})$ becomes an ``unconstrained" system for $(\Pi^{\alpha,1}\dots,\Pi^{\alpha,k-1})\in[0,1]^{k-1}$. 
Henceforth, we shall use this identification whenever convenient.  
\end{remark}

We finish this section with a short lemma. It states that the conditional distribution $\pia$ can never reach the relative border of the simplex $\Delta$, provided it starts from the relative interior. This amounts to saying that all regimes have some positive probability at time zero.
\begin{lemma}
\label{open_simplex}
If $\mu_0\in\Delta^\circ$, then $\pia_t\in\Delta^\circ$ for all $0\leq t\leq T$ a.s.
\end{lemma}

\begin{proof}
See Appendix \ref{proof:open_simplex}.
\end{proof}

In light of the previous lemma, we will assume from here onwards that 
\begin{equation}
\label{assumption_mu_0}
\mu_0\in\Delta^\circ,
\end{equation}
and therefore work with $\Delta^\circ$ instead of $\Delta$.


\section{Value function and HJB equation}
\label{s:viscosity}
\setcounter{assumptions}{0}
\setcounter{subsection}{1}
In this section, we tackle the control problem of the MM with the filter as an additional state variable. We define the MM's value function and we aim to characterize it by means of an HJB equation. 

Let $t\in[0,T]$. We consider our model ``starting at $t$" instead of $0$. Whenever a process is defined from time ``$t^-$ onwards" (i.e., from time $t$ onwards and decreeing its left-limit value at $t$) this implicitly means it is constant on $[0,t)$. 
In particular, we use this convention for all integrals (stochastic or not) of the form $\int_t^\cdot$ and (with a slight abuse of notation) for the processes $W - W_t$ and $\Npm-\Npm_{t^-}$. We work with $t^-$ instead of $t$ due to the jumps of the processes $\Nm,\Np$. However, since $\Nm,\Np$ are quasi-left continuous,\footnote{For example, because they are increasing càdlàg processes admitting continuous compensators for (any one of) the physical probabilities \cite[p.70 Prop.1.19 or p.77 Prop.2.9]{JS}.} for most intended purposes one can drop the left limit with no harm
. We define $\mathbb F^{t,W,N}=(\sigma(W_r-W_t,N_r-N_{t^-}: t\leq r\leq u\vee t)\vee \mathcal F_0)_u$ and $\mathbb F^{t,N}$ analogously. 

Let $s,x\in\mathbb R,\ n\in\mathbb Z\cap [-\Nmax,\Nmax]$ and $\pi\in\Delta^\circ\subset\mathbb R^k$. The set $\mathcal A_t$ of admissible spreads starting at $t$ is the set of $\delta\in\mathcal A$ which are independent of $\mathcal F^{N,W}_{t^-}$ (equivalently, 
the $\delta\in\mathcal A$ which are $\mathbb F^{t,W,N}$-predictable). Consider for each $\alpha=(\delm,\delp)\in\mathcal A^2_t$ the processes $\St, \Xt, \Nt, \pit$ defined pathwise, outside some set $A\in\mathcal F_0$, by (\ref{S}), (\ref{X}), (\ref{N}), (\ref{Pi_equivalent}) resp., replacing the initial conditions $s_0,x_0,n_0,\mu_0$ at time $0$ by $s,x,n,\pi$ resp. at time $t^-$. We remark that $\mathbb F^{t,W,N} = \mathbb F^{t,W,\Nt}$ (since $\Nt = n+\Nm -\Np - (\Nm_{t^-} - \Np_{t^-})$) and all the processes defined in this section are adapted to this filtration. We further assume there exists a family of ``physical" probabilities $(\mPt)_{\alpha\in\mathcal A_t^2}$ such that their null sets generate $\mathcal F_0$, and for $(\mPt,\mathbb F^{t,W,N})$ it holds that $W - W_t$ is a Wiener process and $\Npm-\Npm_{t^-}$ has predictable intensity $\widehat{\lpm}^{\alpha,t,n,\pi}$, as defined in Proposition \ref{observable_intensities} in terms of $\Nt$ and $\pit$. 

We define the penalized P\&L from $t$ to $T$ (see (\ref{problem}) for parameter restrictions) as
\begin{equation}
\label{PnL}
P^{\alpha,s,x,n}_{t,T}\defeq \Xt_T+\St_T \Nt_T- \ell(\Nt_T) - \frac{1}{2}\sigma^2\zeta\int_t^T(\Nt_u)^2du,
\end{equation}
and the \textit{value function} of problem (\ref{problem}) as
\begin{equation}
\label{value_function}
V(t,s,x,n,\pi)\defeq\sup_{\alpha\in \mathcal A^2_t}\mEt\left[U_\gamma\Big(P^{\alpha,s,x,n}_{t,T}\Big)\right].
\end{equation}
Our goal is to compute optimal or ``close to optimal" strategies.\footnote{By ``close to optimal" we mean that for each $\ve>0$ there exists a strategy such that the supremum in (\ref{value_function}) is attained up to $\ve$.} The Dynamic Programming Principle and Ito's Lemma allow us to formally derive (see, e.g., \cite{Bo}) the \textit{Hamilton--Jacobi--Bellman} (or \textit{dynamic programming}) partial-integro differential equation
associated to $V$: 
\begin{equation}
\label{HJB_V}
\begin{split}
0 & = v_t + \mu v_s +  \frac{1}{2}\sigma^2 v_{ss} + \sum_{i,j=1}^k q_t^{ji}\pi^j v_{\pi^i} +  \frac{1}{2}\sigma^2\zeta n^2 (\gamma v -1)\\
  & + \mathbbm 1_{\{n<\Nmax\}}\sup_{\delm\in\overline{(\delmin,\delmax)}}\Big\{\sum_{i,j=1}^k (\Ljm-\Lim)(\delm)\pi^j\pi^i v_{\pi^i} + \Dm_{\delm}(v) \sum_{i=1}^k\pi^i\Lim(\delm)\Big\} \\
	& + \mathbbm 1_{\{-n<\Nmax\}}\sup_{\delp\in\overline{(\delmin,\delmax)}}\Big\{\sum_{i,j=1}^k(\Ljp-\Lip)(\delp)\pi^j\pi^i v_{\pi^i}+ \Dp_{\delp}(v) \sum_{i=1}^k\pi^i\Lip(\delp)\Big\},
\end{split}
\end{equation}
with terminal condition $v(T,s,x,n,\pi) = U_\gamma(x+sn -\ell(n))$, where:

\begin{equation*}
\begin{split}
&\Dpm_{\delta}(v)(t,s,x,n,\pi)\\
        & = v\Big(t,s,x \pm(s \pm \delta),n\mp 1, \frac{1}{\sum_{j=1}^k\pi^j\Ljpm(\delta)}\big(\pi^1\Lpm_1(\delta),\dots,\pi^k\Lpm_k(\delta)\big)\Big)- v(t,s,x,n,\pi),
\end{split}
\end{equation*}
and we convene the following:
\begin{notation*}
\textit{(i)} The derivatives with respect to $\pi=(\pi^1,\dots,\pi^k)$ should be understood via the identification of Remark \ref{identification}. 
\textit{(ii)} Although it is not meaningful to evaluate $v$ on the inventories $\pm\Nmax\pm 1$, this only happens in equation (\ref{HJB_V}) when the corresponding term vanishes. This slight abuse of notation can be found throughout previous works and we will be using it as well. 
\end{notation*}

Equation (\ref{HJB_V}) can also be seen as a coupled system of PIDEs indexed in $n\in\mathbb Z\cap [-\Nmax,\Nmax]$. (We will talk about system of equations or simply ``equation" indistinctly). Nonlinearity aside, (\ref{HJB_V}) is rather complex, in particular due to being of second order, high-dimensional and with derivatives in almost all of these dimensions. Tackling it directly (either analytically or numerically) is utterly challenging. Consequently, it has become common practice for optimal market making and optimal liquidation models \textit{à la} Avellaneda--Stoikov \cite{AS} to propose an ansatz for the solution \cite{AS,BL,CDJ,CJ,CJR,FL1,FL2,G,GL,GLFT}. This approach however, relies heavily on the existence of a classical solution for the resulting simplified equation, so that the ansatz is ultimately proved valid by a suitable verification theorem. (See Section \ref{s:full_info} for more details.) When the simplified equation does not admit (or cannot be guaranteed to admit) a classical solution, and a viscosity approach needs to be used instead, the previous argument breaks down.

If we attempted to solve our problem by the standard approach,
 a plausible ansatz for the value function
 could be
\begin{equation}
\label{decomposition}
V(t,s,x,n,\pi) = U_\gamma(x+sn + \Theta(t,n,\pi)).
\end{equation}
\noindent Formal substitution 
yields the following equation for $\Theta$:

\begin{equation}
\label{HJB_Theta}
\begin{split}
0 & = \theta_t + \mu n - \frac{1}{2}\sigma^2 n^2 (\gamma+\zeta)+ \sum_{i,j=1}^k q_t^{ji}\pi^j \theta_{\pi^i}\\
  & + \mathbbm 1_{\{n<\Nmax\}}\sup_{\delm\in\overline{(\delmin,\delmax)}}\Big\{\sum_{i,j=1}^k (\Ljm-\Lim)(\delm)\pi^j\pi^i \theta_{\pi^i} + U_\gamma\big(\delm+\Dm_{\delm}(\theta)\big) \sum_{i=1}^k\pi^i\Lim(\delm)\Big\}\\
	& + \mathbbm 1_{\{-n<\Nmax\}}\sup_{\delp\in\overline{(\delmin,\delmax)}}\Big\{\sum_{i,j=1}^k(\Ljp-\Lip)(\delp)\pi^j\pi^i \theta_{\pi^i}+ U_\gamma\big(\delp+\Dp_{\delp}(\theta)\big) \sum_{i=1}^k\pi^i\Lip(\delp)\Big\},
\end{split}
\end{equation}
with terminal condition $\theta(T,n,\pi) = -\ell(n)$, where:
\begin{equation*}
\begin{split}
\Dpm_{\delta}(\theta)(t,n,\pi) = \theta\Big(t,n\mp 1, \frac{1}{\sum_{j=1}^k\pi^j\Ljpm(\delta)}\big(\pi^1\Lpm_1(\delta),\dots,\pi^k\Lpm_k(\delta)\big)\Big) - \theta(t,n,\pi).
\end{split}
\end{equation*}
The new system of PIDEs is of first order and no longer depends on the variables $s$ and $x$ (there is no diffusion anymore). This is a considerable simplification; one that will permit effective numerical solution in Section \ref{s:numerics}. But it is not good enough for us to assert existence of a classical solution. Notwithstanding, we are able to rigorously prove the decomposition (\ref{decomposition}) and explicitly find $\Theta$ as a new ``value function" (Theorem \ref{main_theorem_1}). When the control space is compact, this ultimately allows us to characterize $\Theta$ as the unique solution of the terminal condition PIDE (\ref{HJB_Theta}) in the viscosity sense (Theorem \ref{main_theorem_2}), further simplified in the unconstrained inventory case. These two theorems constitute the main mathematical results of this paper. They allow us to safely postulate reasonable candidates for optimal (or $\epsilon$-optimal) strategies for the MM, i.e., those given by spreads that (at least approximately) realize the suprema in (\ref{HJB_Theta}).

\begin{theorem}
\label{main_theorem_1}
There exists a unique function $\Theta:[0,T]\times(\mathbb Z\cap[-\Nmax,\Nmax])\times\Delta^\circ\to\mathbb R$ such that the decomposition (\ref{decomposition}) holds true. Furthermore, 
there exists a family of equivalent probability measures $\tilde{\mathbb P}^{\alpha,t,n,\pi}\sim\mPt$, $\alpha=(\delm,\delp)\in\mathcal A_t^2$, such that
\begin{enumerate}[(i)]
\item $\pit$ is the unique strong solution of (\ref{Pi_equivalent}) with initial condition $(t^-,\pi)$ under $\tilde{\mathbb P}^{\alpha,t,n,\pi}$.
\item $\mu^{\Nt}$ has $(\tilde{\mathbb P}^{\alpha,t,n,\pi},\mathbb F^{t,W,N})$-predictable intensity kernel 
$$\tilde{\eta}_u^N(dz) \defeq e^{-\gamma\delm_u}\hlmtu m_{1}(dz) + e^{-\gamma\delp_u}\hlptu m_{-1}(dz).$$
\item $\Theta=U_\gamma^{-1}\circ\Psi = -\frac{1}{\gamma}\log(1-\gamma\Psi)$ with 
\begin{equation}
\label{Psi}
\Psi(t,n,\pi)\defeq\sup_{\alpha\in\tilde{\mathcal A}^2_t}\tilde{\mathbb E}^{\alpha,t,n,\pi}\left[U_\gamma\Big(\tilde P_{t,T}^{\alpha,n,\pi}\Big)\right],
\end{equation}
\end{enumerate}
where $\tilde P_{t,T}^{\alpha,n,\pi}\defeq\int_t^T\big\{ U_\gamma(\delm_u) \hlmtu + U_\gamma(\delp_u) \hlptu + \mu \Nt_u -\frac{1}{2}\sigma^2(\gamma+\zeta)(\Nt_u)^2\big\}du - \ell(\Nt_T)$ and $\tilde{\mathcal A}_t\defeq\{\delta\in\mathcal A:\delta\mbox{ is }\mathbb F^{t,\Nt}\mbox{-predictable}\}$.
\end{theorem}

\begin{proof}
For shortness, we make an abuse of notation and omit $(t,n,\pi)$ from the probability measures and expectations. Let us start by proving (\ref{decomposition}) and finding $(\tilde{\mathbb P}^\alpha)_{\alpha\in\mathcal A_t^2}$ with the desired properties.  
Using integration by parts we can re-write the penalized P\&L (\ref{PnL}) as

\begin{equation}
\label{Pbar}
\begin{split}
P^{\alpha,s,x,n}_{t,T}& = x +sn + \int_t^T\left\{\mu\Nt_u -\frac{1}{2}\sigma^2(\gamma+\zeta)(\Nt_u)^2\right\}du - \ell(\Nt_T)\\
					 &  +\sigma\int_t^T \Nt_u dW_u +  \frac{1}{2}\sigma^2\gamma\int_t^T (\Nt_u)^2du +\int_t^T \delm_u d\Nm_u + \delp_u d\Np_u \\
					& \backdefeq x+ sn + \overline P_{t,T}^{\alpha,n},
\end{split}
\end{equation}
Consider first the case $\gamma=0$. The integrals with respect to $W,\Nm,\Np$ all have bounded integrands (and predictable for $\Nm,\Np$), except in the case of unconstrained inventory: $\Nmax=+\infty$ and $\gamma=\zeta=0\equiv\ell$ (see (\ref{problem})). Regardless, we still have
$
\mEa\left[\int_t^T (\Nt_u)^2 du\right]=\int_t^T \mEa\left[(\Nt_u)^2\right] du<+\infty\mbox{ by (\ref{bounded_moments})}.
$
Choosing $\tilde{\mathbb P}^\alpha\defeq\mPa$ the conclusion follows by taking expectation and by Propositions \ref{observable_intensities} and \ref{proposition_Pi}. 

Consider now $\gamma>0$. Hence, we are in the case $|N|\leq\Nmax<+\infty$. We define 
$$
A^\alpha\defeq\mathcal E\left( -\gamma\sigma\int_t^\cdot \Nt_udW_u\right)=\exp\left(-\gamma\sigma\int_t^\cdot \Nt_ud W_u-  \frac{1}{2}\sigma^2\gamma^2\int_t^\cdot (\Nt_u)^2du \right).
$$
By Novikov's condition, $A^\alpha$ is a strictly positive uniformly integrable (UI) martingale with $\mEa[A_T]=1$, and therefore defines an equivalent probability measure $\mathbb A^\alpha\sim\mPa$ via $\frac{d\mathbb A^\alpha}{d\mPa}=A^\alpha_T$. Note that the Girsanov--Meyer Theorem. \cite[p.132 Thm.35]{P} ensures the $\mathbb F^{t,W,N}$-intensities of $\Npm - \Npm_{t^-}$ remain the same when changing to $\mathbb A^\alpha$. Let us set 
\begin{equation*}
\begin{split}
B^\alpha & \defeq\mathcal E\left(-\gamma\int_t^\cdot U_\gamma(\delm_u)d\overline{\Nm_u}^\alpha + U_\gamma(\delp_u)d\overline{\Np_u}^\alpha \right)\\
				 & =\mathcal E\left(\int_t^\cdot \big(e^{-\gamma\delm_u} -1\big)d\overline{\Nm_u}^\alpha + \big(e^{-\gamma\delp_u} -1\big)d\overline{\Np_u}^\alpha\right),
\end{split}
\end{equation*}
where $\overline{\Npm_u}^\alpha$ denote the corresponding $(\mPa,\mathbb F^{t,W,N})$-compensated (or equivalently, $(\mathbb A^\alpha,\mathbb F^{t,W,N})$-compensated) processes. By the same arguments of Propositions \ref{Girsanov1} and \ref{Girsanov2}, $B^\alpha$ is a strictly positive UI martingale with $\mathbb E^{\mathbb A^\alpha}[B^\alpha_T]=1$ and defines an equivalent probability measure $\tilde{\mathbb P}^\alpha\sim\mathbb A^\alpha\sim\mPa$ via $\frac{d\tilde{\mathbb P}^\alpha}{d\mathbb A^\alpha}=B^\alpha_T$, such that \textit{(ii)} holds true. Note that \textit{(i)} is also trivially verified due to the equivalence of the probability measures. 

Suppose for the time being that $\Psi$ is defined as in (\ref{Psi}) but taking supremum over the whole set of admissible controls $\mathcal A_t^2\supseteq \tilde{\mathcal A}_t^2$ instead. We will see afterwards that this makes no difference. To see (\ref{decomposition}), observe that the identity $U_\gamma(a+b)=U_\gamma(b)e^{-\gamma a}+U_\gamma(a)$ and (\ref{Pbar}) yield 
$
U_\gamma(P^{\alpha,s,x,n}_{t,T}) = U_\gamma(\overline P^{\alpha,n}_{t,T})e^{-\gamma(x+sn)} + U_\gamma(x+sn),
$
giving
$$
V(t,s,x,n,\pi)= \sup_{\alpha\in\mathcal A_t^2}\mEa\left[ U_\gamma(\overline P^{\alpha,n,\pi}_{t,T})\right]e^{-\gamma(x+sn)} + U_\gamma(x+sn).
$$
On the other hand, by the same identity, 
$$
U_\gamma(x+sn+\Theta)= (U_\gamma\circ\Theta) e^{-\gamma(x+sn)} + U_\gamma(x+sn)= \Psi e^{-\gamma(x+sn)} + U_\gamma(x+sn).
$$
As a consequence, (\ref{decomposition}) is equivalent to the equality $\Psi(t,n,\pi)=\sup_{\alpha\in\mathcal A_t^2}\mEa\left[ U_\gamma(\overline P^{\alpha,n}_{t,T})\right]$. We check instead the stronger statement
\begin{equation}
\label{aux}
\mEa\left[\exp\Big(-\gamma\overline P^{\alpha,n}_{t,T} \Big) \right]= \tilde{\mathbb E}^{\alpha}\left[\exp\Big(-\gamma \tilde P_{t,T}^{\alpha,n} \Big) \right],\mbox{ for all }\alpha\in\mathcal A_t^2.
\end{equation} 
Using the explicit exponential formula (see (\ref{exponential_explicit})) and by straightforward computations:
\begin{equation*}
\begin{split}
\exp\Big(-\gamma\overline P^{\alpha,n}_{t,T} \Big) & = \exp\left(-\gamma \Big(\int_t^T\left\{\mu\Nt_u -\frac{1}{2}\sigma^2(\gamma+\zeta)(\Nt_u)^2\right\}du - \ell(\Nt_T) \Big)\right)\\
& \times A^\alpha_T\exp\left(-\gamma\int_t^T \delm_u d\Nm_u + \delp_u d\Np_u\right)\\
& = \exp\Big(-\gamma \tilde P_{t,T}^{\alpha,n} \Big)A^\alpha_T
\exp\left(\gamma\int_t^T\big\{ U_\gamma(\delm_u) \hlmtu + U_\gamma(\delp_u) \hlptu\big\}du\right)\\
& \times \prod_{\substack{t\leq u \leq T:\\ \Delta\Nm_u\neq 0}}\exp\left(-\gamma\delm_u \right)\prod_{\substack{t\leq u \leq T:\\ \Delta\Np_u\neq 0}}\exp\left(-\gamma\delp_u \right) \\
& = \exp\Big(-\gamma \tilde P_{t,T}^{\alpha,n}\Big)A^\alpha_T B^\alpha_T,
\end{split}
\end{equation*}
which yields (\ref{aux}) after taking $\mPa$-expectation.

It remains to see that 
$$\Psi(t,n,\pi)\defeq\sup_{\alpha\in\mathcal A^2_t}\tilde{\mathbb E}^{\alpha}\left[U_\gamma\Big(\tilde P_{t,T}^{\alpha,n,\pi}\Big)\right] = \sup_{\alpha\in\tilde{\mathcal A}^2_t}\tilde{\mathbb E}^{\alpha}\left[U_\gamma\Big(\tilde P_{t,T}^{\alpha,n,\pi}\Big)\right]\backdefeq \tilde{\Psi}(t,n,\pi).$$
Clearly $\Psi\geq \tilde{\Psi}$. Let us check $\Psi\leq\tilde{\Psi}$. As done in Proposition \ref{Girsanov2} we can define a family of ``reference" equivalent probability measures $\tilde{\mathbb Q}^\alpha\sim\tilde{\mathbb P}^\alpha$, $\alpha\in\mathcal A^2_t$, such that for $(\tilde{\mathbb Q}^\alpha,\mathbb F^{t,W,N})$ it holds: $W - W_t$ is a Wiener process independent of the counting process $(\Nm-\Nm_{t^-},\Np-\Np_{t^-})$ and $\Npm-\Npm_{t^-}$ has predictable intensity $\mathbbm 1^{\pm,t,n}_u\defeq\mathbbm 1_{\{\mp \Nt_{u^-}<\Nmax\}}$ (in particular, its law does not depend on $\alpha$). Furthermore, the inverse change of measure is given by $d\tilde{\mathbb P}^\alpha/d\tilde{\mathbb Q}^\alpha=\tilde{Z}^\alpha_T$ with
$$
\tilde{Z}^\alpha\defeq\mathcal E\left(\int_t^\cdot \big(\widehat{\Lm}(\pit_{u^-},\delm_u)-1\big)\big(d\Nm_u- \mathbbm 1^{-,t,n}_u du\big)+\big(\widehat{\Lp}(\pit_{u^-},\delp_u)-1\big)\big(d\Np_u- \mathbbm 1^{+,t,n}_u du\big)\right).
$$
Let us fix $\alpha\in\mathcal A_t^2$. Denote by $\mathcal D=\mathcal D([t,T],\mathbb R)$ the Skorokhod space of càdlàg functions with its usual sigma algebra and by $\mathbb P^W, \mathbb P^N$ the laws (or pushforward measures) induced on $\mathcal D$ by $W-W_t,\Nt$ resp. when starting from $\mathbb Q^\alpha$. These laws do not depend on $\alpha$ and characterize the joint law of $(W-W_{t},\Nt)$ on $\mathcal D^2$ as $\mathbb P^W\otimes\mathbb P^N$, due to the independence of the two processes. Since $\alpha$ is $\mathbb F^{t,W,N}$-predictable, by a monotone class argument
one can show there exists a jointly measurable process $f:[t,T]\times\mathcal D^2\to\mathbb R$ such that $\alpha_u=f_u(W-W_t, \Nt)$ and for $\mathbb P^W$-almost every $w\in\mathcal D$, the process $\tilde{\alpha}_u\defeq f_u(w, \Nt)$ is in $\tilde{\mathcal A}_t^2$. Note also that we can write $\tilde{Z}^\alpha_T U_\gamma\Big(\tilde P_{t,T}^{\alpha,n,\pi}\Big)=g(\alpha, \Nt) = g(f(W-W_t,\Nt),\Nt)$ for some function $g$.
By Fubini's theorem,
\begin{equation*}
\begin{split}
\tilde{\mathbb E}^{\alpha}\left[U_\gamma\Big(\tilde P_{t,T}^{\alpha,n,\pi}\Big)\right] = \tilde{\mathbb E}^{\mathbb Q^\alpha}\left[\tilde{Z}^\alpha_T U_\gamma\Big(\tilde P_{t,T}^{\alpha,n,\pi}\Big)\right] = \int \mathbb E^{\mathbb P^N}\left[g(f(w,\cdot),\cdot)\right]d\mathbb P^W(w) \leq \int\tilde{\Psi} d\mathbb P^W(w)=\tilde{\Psi}.
\end{split} 
\end{equation*}
Since $\alpha\in\mathcal A_t^2$ was arbitrary, we conclude that $\Psi=\tilde{\Psi}$.
\end{proof}
 
Just as it occurs under full information (see Section \ref{s:full_info}), for a fully risk-neutral MM with negligible costs (i.e. $\Nmax=+\infty,\ \gamma=\zeta=0\equiv\ell$), $\Theta$ can be further decomposed. Note, from their definition, that in this case $\hlpmt$ and $\pit$ do not depend on $n$. As it was proved in Theorem \ref{main_theorem_1}, when $\gamma=0$ the family $(\tilde{\mathbb P}^{\alpha,t,n})$ can be taken as the original physical probabilities, and these do not depend on $n$ either. The following corollary is now immediate. 
\begin{corollary}
\label{coro_main_theorem_1}
If $\Nmax=+\infty$ and $\gamma=\zeta=0\equiv\ell$ then 
$$V(t,s,x,n,\pi) = x+sn+\mu n(T-t)+\Phi(t,\pi),$$ 
with $\Phi(t,\pi)=\sup_{\alpha\in\tilde{\mathcal A}_t^2}\mathbb E^{t,\pi}\left[\int_t^T \big\{\delm_u\widehat{\lm_u}^{\alpha,t,\pi} + \delp_u\widehat{\lp_u}^{\alpha,t,\pi} + \mu (\widehat{\lm_u}^{\alpha,t,\pi} - \widehat{\lp_u}^{\alpha,t,\pi})\big\}du\right].$
\end{corollary}

In the context of the previous corollary, formal substitution in (\ref{HJB_V}) or (\ref{HJB_Theta}) yields the following PIDE for $\Phi$:
\begin{equation}
\label{HJB_Phi}
\begin{split}
0 & = \phi_t + \sum_{i,j=1}^k q_t^{ji}\pi^j \phi_{\pi^i} \\
  & + \sup_{\delm\in\overline{(\delmin,\delmax)}}\Big\{\sum_{i,j=1}^k (\Ljm-\Lim)(\delm)\pi^j\pi^i \phi_{\pi^i} + \big(\delm+\mu(T-t)+\Dm_{\delm}(\phi)\big) \sum_{i=1}^k\pi^i\Lim(\delm)\Big\}\\
	& + \sup_{\delp\in\overline{(\delmin,\delmax)}}\Big\{\sum_{i,j=1}^k(\Ljp-\Lip)(\delp)\pi^j\pi^i \phi_{\pi^i}+ \big(\delp-\mu(T-t)+\Dp_{\delp}(\phi)\big) \sum_{i=1}^k\pi^i\Lip(\delp)\Big\},
\end{split}
\end{equation}
with terminal condition $\phi(T,\pi) = 0$, where:
\begin{equation*}
\begin{split}
\Dpm_{\delta}(\phi)(t,\pi) = \phi\Big(t, \frac{1}{\sum_{j=1}^k\pi^j\Ljpm(\delta)}\big(\pi^1\Lpm_1(\delta),\dots,\pi^k\Lpm_k(\delta)\big)\Big) - \phi(t,\pi).
\end{split}
\end{equation*}

\begin{remark}
Besides its theoretical motivation, Theorem \ref{main_theorem_1} is interesting in itself. It enables the use of probabilistic and PDMPs numerical techniques when the dimension of the problem renders PDEs methods prohibitive. Furthermore, as readily seen in the proof, the decomposition of the expected P\&L utility holds true for any admissible strategy. This means, in particular, that one can efficiently simulate the proceeds of any strategy without the need to simulate the diffusive state variables (see Section \ref{section_simulations}).
\end{remark}

We now want to prove that $\Theta$ (resp. $\Phi$) is the unique continuous viscosity solution of the terminal condition PIDE (\ref{HJB_Theta}) (resp. (\ref{HJB_Phi})). (See, e.g., \cite[Def.2.1]{Son} for the relevant definition, or more in general \cite[Def.7.3]{DF}, recalling that in our case we have no boundary conditions other than that at terminal time.) A complication inevitably arises, as classical viscosity techniques \cite{Bo,FS,OS} cannot be applied directly to weak formulation models such as ours. However, the decomposition of Theorem \ref{main_theorem_1} (resp. Corollary \ref{coro_main_theorem_1}) not only reduces the dimensionality of the problem, but also states that the MM may neglect the diffusion component of the state process altogether, focusing solely on the time-space state variable $(u,\Nt,\pit)$ (resp. $(u,\pit)$). This is a PDMP as introduced in \cite{D1} (detailed treatments also found in \cite{BR3, D2}). 
Using results from PDMPs theory, our continuous-time problem is identified with a control problem for a discrete-time Markov decision model, as in \cite{BR1,BR2,BR3,CEFS}, and linked again with viscosity solutions of HJB PIDEs as in \cite{CEFS,DF}. 

An inevitable drawback is that the PDMPs approach relies on the use of the so-called \textit{randomized} (or \textit{relaxed}) controls and requires the control space to be compact. Hence, for the following theorem we will assume $-\infty<\delmin<\delmax<+\infty$. Assuming a uniform lower constraint $\delmin>-\infty$ is hardly a problem. On the contrary, $\delmin=0$ (or even some small positive number) is the most meaningful in practice, as negative spreads imply the MM is willing to offer her clients better prices than the reference price $S$. ($\delmin=-\infty$ is motivated in the literature by mathematical convenience rather than modelling accuracy.) A uniform upper bound $\delmax<+\infty$, on the other hand, is harder to assess a priori.
 
Fortunately, in most situations encountered in practice, the unconstrained optimization will yield bounded optimal spreads nonetheless, and the MM can dispense with $\delmin,\delmax$ if she wishes to do so (see Section \ref{s:numerics} for an example). 

\begin{theorem}
\label{main_theorem_2}
Assume $-\infty<\delmin<\delmax<\infty$. For $\Nmax<\infty$ (resp. $\Nmax=+\infty$ and $\gamma=\zeta=0\equiv\ell$), let $\Theta$ (resp. $\Phi$) be as in Theorem \ref{main_theorem_1} (resp. Corollary \ref{coro_main_theorem_1}). Then $\Theta$ (resp. $\Phi$) is the unique continuous viscosity solution of the terminal condition PIDE (\ref{HJB_Theta}) (resp. (\ref{HJB_Phi})). 
\end{theorem}

\begin{proof}
See Appendix \ref{proof:main_theorem2}.
\end{proof}

\section{Full information}
\label{s:full_info}
\setcounter{assumptions}{0}
\setcounter{subsection}{0}

In this section we consider the idealized case of a MM with full information. We assume the MM has inside information in such a way that she can observe the full filtration $\mathbb F$, and in particular she can observe $Y$. For example, if $Y$ represents different levels of competition amongst 
liquidity providers, this would practically mean the MM has information regarding her competitors' quotes. We will see that in this case the value function turns out to be a regular, classical solution, of its HJB equation. Afterwards, we will compare the results with those obtained in the more realistic setting of partial information.

\subsection{Dimensionality reduction and the general system of ODEs}
We consider problem (\ref{problem}) but under full information, with the set of admissible spreads: 
$$\mathcal U\defeq\{\delta:[0,T]\to\overline{(\delmin,\delmax)}: \delta\mbox{ is }\mathbb F\mbox{-predictable and bounded from below}\}.$$
Note that in this section, and when $\delmax=+\infty$, we do not assume upper-boundedness of the spreads. Let $t\in [0,T],\ s,x\in\mathbb R$ and $(n,i)\in\mathcal I\defeq\left(\mathbb Z\cap [-\Nmax,\Nmax]\right)\times \{1,\dots,k\}$. Consider the processes $\St,\Xt,\Nt,P^{\alpha,s,x,n}$ as defined in Section \ref{s:viscosity} and $\Yt$ a Markov chain with deterministic generator matrix $Q$, state space $\{1,\dots,k\}$ and such that $\Yt_{t^-}=i$. We assume the physical probabilities $\mPtfull$ are defined for every $\alpha\in\mathcal U^2$ and that Assumptions \ref{assumptions_Q}, \ref{assumptions_intensities} and (\ref{no_common_jumps}) (starting at time $t^-$) are still in place. The value function in this case is 
\begin{equation}
\label{value_function_full_info}
V(t,s,x,n,i)\defeq \sup_{\alpha\in\mathcal U_t^2}\mathbb E^{t,n,i}\Big[U_\gamma\big(P^{\alpha,s,x,n}_{t,T}\big)\Big],
\end{equation}

By means of the Dynamic Programming Principle and Ito's Lemma one can formally derive the following HJB PIDE for $V$:
\begin{equation}
\label{HJB_V_full_info}
\begin{split}
0 & = v_t(t,s,x,n,i) + \mu v_s(t,s,x,n,i) +  \frac{1}{2}\sigma^2 v_{ss}(t,s,x,n,i)\\
  & + \frac{1}{2}\sigma^2\zeta n^2 (\gamma v(t,s,x,n,i) -1)+ \sum_{j=1}^k q_t^{ij}v(t,s,x,n,j)\\
  & + \mathbbm 1_{\{n<\Nmax\}}\sup_{\delm\in\overline{(\delmin,\delmax)}}\Lim(\delm)\Big(v\big(t,s,x -(s - \delm),n+1, i\big)- v(t,s,x,n,i)\Big)\\
	& + \mathbbm 1_{\{-n<\Nmax\}}\sup_{\delp\in\overline{(\delmin,\delmax)}}\Lip(\delp)\Big(v\big(t,s,x +(s + \delp),n-1, i\big)- v(t,s,x,n,i)\Big),
\end{split}
\end{equation}
with terminal condition $v(T,s,x,n,i) = U_\gamma(x+sn -\ell(n))$. 

In this new context, instead of formally proving a decomposition of $V$ as in Theorem \ref{main_theorem_1}, it is more straightforward to propose an ansatz and ultimately prove it valid with a verification theorem. (This is the standard approach used in the Avellaneda--Stoikov framework.) Let us consider an ansatz for the value function analogous to those used for the one regime case:
$$
V(t,s,x,n,i) = U_\gamma(x+sn + \Theta(t,n,i)),
$$
for some function $\Theta:[0,T]\times\mathcal I\to\mathbb R$, $C^1$ in time. Substituting in (\ref{HJB_V_full_info})
and using Assumptions \ref{assumptions_Q}, we see that $\Theta$ must satisfy 
a system of ODEs indexed in $(n,i)\in\mathcal I$:
\begin{equation}
\label{HJB_Theta_full_info}
\begin{split}
0 & = \theta_t(t,n,i) + \mu n - \frac{1}{2}\sigma^2 n^2(\zeta+\gamma)+ \sum_{j\neq i} q_t^{ij} U_\gamma\big(\theta(t,n,j)-\theta(t,n,i)\big)\\
  & + \mathbbm 1_{\{n<\Nmax\}}\Hm_i\big(\theta(t,n+1,i)-\theta(t,n,i)\big)+ \mathbbm 1_{\{-n<\Nmax\}}\Hp_i\big(\theta(t,n-1,i)-\theta(t,n,i)\big),
\end{split}
\end{equation}
with terminal condition $\theta(T,n,i) = -\ell(n)$, where:
\begin{enumerate}
\item $\Hpm_i(d)\defeq\sup_{\delta\in\overline{(\delmin,\delmax)}}\hpm_i(\delta, d)$, for $d\in\mathbb R$.
\item $\hpm_i(\delta, d)\defeq \Lipm(\delta)U_\gamma(\delta+ d)$, for $d\in\mathbb R,\ \delta\in\overline{(\delmin,\delmax)}$.
\end{enumerate}

The original problem is simplified in this way, both by the dimension of the state variable and by the complexity of the equations, 
provided we can show that problem (\ref{HJB_Theta_full_info}) admits a solution. With this aim in mind, let us prove first the following property of the Hamiltonian functions $\Hpm_1,\dots,\Hpm_k$. 

\begin{lemma}
\label{lipschitz_hamiltonians}
For all $1\leq i\leq k$, it holds:
\begin{enumerate}[(i)]
\item For each compact $K\subset\mathbb R$, there exists $[a,b]\subseteq\overline{(\delmin,\delmax)}$ such that $\Hpm_i(d)=\max_{\delta\in [a,b]}\hpm_i(\delta, d)$, for all $d\in K$. 
\item $\Hpm_i$ is locally Lipschitz.
\end{enumerate}
\end{lemma}

\begin{proof}
Fix $1\leq i\leq k$ and let $K\subset\mathbb R$ be a compact set. We verify first that \textit{(i)} is a consequence of Assumptions \ref{assumptions_intensities}. Let $C>0$ such that $|d|\leq C$ for all $d\in K$. If $\delmin=-\infty$, then we can take any $a<\min\{-C,\delmax\}$ since $\hpm_i(\delta,d)<\hpm_i(a,d)$ for all $\delta<a$ and $d\in K$. On the other hand, if $\delmax=+\infty$, take some $c>\max\{\delmin,2C\}$. It holds that $\hpm_i(c,d)\geq\Lpm_i(c)U_\gamma(C)\backdefeq\varepsilon>0$ and we can choose $b>c$ such that $\hpm_i(\delta,d)\leq\hpm_i(\delta,C)<\varepsilon$ for all $\delta\geq b$ and $d\in K$. Replacing supremum by maximum is now immediate due to the continuity of $\hpm_i(\cdot,d)$ on $[a,b]$ for all $d$.

\textit{(ii)} is routinely verified using that the family $\{\hpm_i(\delta,\cdot)\}_{\delta\in[a,b]}$ is equi-Lipschitz on $K$. 
\end{proof}

We want to prove now that the Cauchy problem (\ref{HJB_Theta_full_info}) admits a unique global classical solution $\theta$ which is $C^1$ in time. To this purpose, we will treat the cases of the finite system ($\Nmax<\infty$) and infinite system ($\Nmax=\infty$) of equations separately.

\subsection{Constrained inventory ODEs}
For $\Nmax<\infty$ we are dealing with a finite system of ODEs. We know that under certain regularity conditions the Cauchy problem (\ref{HJB_Theta_full_info}) is guaranteed to have a classical solution on some neighbourhood $(\tau,T]\subset[0,T]$ of $T$. 
Nonetheless, it is not always the case that such a local solution can be extended to a global one on $[0,T]$. Following \cite{G, GL}, we start by proving a comparison principle for (\ref{HJB_Theta_full_info}) that will allow us, in particular, to show the existence of a global solution. The argument used is standard for comparison principles of HJB-type equations. 

\begin{proposition}[\textbf{Comparison Principle}]
\label{comparison_ODEs}
Let $I\subseteq [0,T]$ be an interval containing $T$ and let $\overline\theta,\underline\theta:I\times \mathcal I\to\mathbb R$ be classical ($C^1$ with respect to time) super- and sub-solutions resp. of (\ref{HJB_Theta_full_info}). That is,
\begin{equation}
\label{HJB_Theta_full_info_supersol}
\begin{split}
0 & \leq -\overline\theta_t(t,n,i) - \mu n + \frac{1}{2}\sigma^2 n^2(\zeta+\gamma) - \sum_{j\neq i} q_t^{ij} U_\gamma\big(\overline\theta(t,n,j)-\overline\theta(t,n,i)\big)\\
  & - \mathbbm 1_{\{n<\Nmax\}}\Hm_i\big(\overline\theta(t,n+1,i)-\overline\theta(t,n,i)\big) - \mathbbm 1_{\{-n<\Nmax\}}\Hp_i\big(\overline\theta(t,n-1,i)-\overline\theta(t,n,i)\big),
\end{split}
\end{equation} 
\begin{equation}
\label{HJB_Theta_full_info_subsol}
\begin{split}
0 & \geq -\underline\theta_t(t,n,i) - \mu n + \frac{1}{2}\sigma^2 n^2(\zeta+\gamma) - \sum_{j\neq i} q_t^{ij} U_\gamma\big(\underline\theta(t,n,j)-\underline\theta(t,n,i)\big)\\
  & - \mathbbm 1_{\{n<\Nmax\}}\Hm_i\big(\underline\theta(t,n+1,i)-\underline\theta(t,n,i)\big) - \mathbbm 1_{\{-n<\Nmax\}}\Hp_i\big(\underline\theta(t,n-1,i)-\underline\theta(t,n,i)\big)
\end{split}
\end{equation} 
and
\begin{equation}
\label{terminal_cond_comparison}
\underline\theta(T,\cdot,\cdot)\leq-\ell\leq\overline\theta(T,\cdot,\cdot).
\end{equation}
Then
$$ \underline\theta\leq\overline\theta.$$
\end{proposition}
\begin{proof}
See Appendix \ref{proof:comparison_ODEs}.
\end{proof}
We can now prove the existence and uniqueness of a classical global solution to the Cauchy problem (\ref{HJB_Theta_full_info}). 

\begin{theorem}
\label{Theta_constrained_inventory}
There exists a unique $\Theta:[0,T]\times\mathcal I\to\mathbb R$, $C^1$ in time, which (classically) solves the Cauchy problem (\ref{HJB_Theta_full_info}).
\end{theorem}

\begin{proof}
For each $1\leq i\leq k$, Lemma \ref{lipschitz_hamiltonians} tells us that $\Hpm_i:\mathbb R\to\mathbb R$ is locally Lipschitz. We also assumed $Q:[0,T]\to\mathbb R^{k\times k}$ continuous (Assumptions \ref{assumptions_Q}). Hence, by the Cauchy-Lipschitz Theorem, the terminal condition system of ODEs (\ref{HJB_Theta_full_info}) admits a unique $C^1$ local solution $(\Theta(\cdot,n,i))_{(n,i)\in\mathcal I}$, defined on some maximal interval $I\subset [0,T]$ containing $T$.

Suppose that $I\subsetneq [0,T]$. Then $I=(\tau,T]$ for some $0\leq\tau<T$ and $\|\Theta(t)\|\to\infty$ as $t\searrow\tau$. 
We claim that $\Theta$ is actually bounded, resulting in a contradiction. Indeed, take 
$$\overline K=\max_{(i,n)\in\mathcal I}\left\{\Big|\mu n -\frac{1}{2}\sigma^2 n^2(\zeta+\gamma)+\mathbbm 1_{\{n<\Nmax\}}\Hm_i(0)+ \mathbbm 1_{\{-n<\Nmax\}}\Hp_i(0)\Big|\right\}$$
and
$$
\underline K=\ell(\Nmax),
$$
and define $\overline\theta:I\times\mathcal I\to\mathbb R$ by $\overline\theta(t,n,i)=\overline K(T-t)$ and $\underline\theta\defeq -\overline\theta -\underline K$. Then $\overline\theta$ (resp. $\underline\theta$) is a super- (resp. sub-) solution of the Cauchy problem (\ref{HJB_Theta_full_info}). By the Comparison Principle (\ref{comparison_ODEs}), 
$$
-\overline K T - \underline K\leq\underline\theta\leq\Theta\leq\overline\theta\leq\overline K T,
$$
proving that $\Theta$ is bounded.
\end{proof}

\subsection{Unconstrained inventory ODEs} 
We consider now $\Nmax=+\infty$, for which (\ref{HJB_Theta_full_info}) becomes an infinite system of ODEs. Recall that in this case we assumed $\gamma=\zeta=0\equiv\ell$, i.e., the MM is fully risk-neutral and has negligible costs. This allows us to further reduce the dimensionality of the state variable by the additional ansatz
\begin{equation}
\label{reduction_Gamma}
\Theta(t,n,i)=\mu n(T-t)+\Phi_i(t),
\end{equation}
for some $\Phi=(\Phi_i)_{i=1}^k\in C^1([0,T],\mathbb R^k)$. Substituting in (\ref{HJB_Theta_full_info}), we get that $\Phi$ must solve the finite linear system of ODEs
\begin{equation}
\label{HJB_Phi_full_info}
\begin{cases}
\phi'(t) =-Q(t)\phi(t) + b(t)\\
\phi(T) =0,
\end{cases}
\end{equation}
with $b_i(t)= - \Hm_i(\mu(T-t))-\Hp_i(-\mu(T-t)),\ i=1,\dots,k$. 
By continuity of $Q$ and $b$ (see Lemma \ref{lipschitz_hamiltonians}), the previous system is known to have a unique global solution $\Phi\in C^1([0,T],\mathbb R^k)$ 
which can be computed by the variation of parameters method. Straightforward verification now gives the following:
\begin{proposition}
\label{Theta_unconstrained_inventory}
Let $\Phi\in C^1([0,T],\mathbb R^k)$ be the unique solution of (\ref{HJB_Phi_full_info}). Then $\Theta:[0,T]\times \mathcal I\to\mathbb R$ such that $\Theta(t,n,i)=\mu n(T-t)+\Phi_i(t)$, is the unique, $C^1$ in time, (classical) solution of the Cauchy problem (\ref{HJB_Theta_full_info}) with $\Nmax=+\infty$ and $\gamma=\zeta=0\equiv\ell$.
\end{proposition}

\subsection{General Verification Theorem}
We give now the complete solution for the general model under full information. For the next theorem we note that given the function $\Theta$ defined in Theorem \ref{Theta_constrained_inventory} for $\Nmax<+\infty$ (resp. Proposition \ref{Theta_unconstrained_inventory} for $\Nmax=+\infty$ and $\gamma=\zeta=0\equiv\ell$) the difference or ``jump" terms of equation (\ref{HJB_Theta_full_info}) are bounded by continuity. That is, $\Theta(t,n\mp 1,i)-\Theta(t,n,i)$ is bounded on $[0,T]\times\mathcal I$ for $\Nmax<+\infty$ (resp. $\Theta(t,n\mp 1,i)-\Theta(t,n,i)=\mp\mu(T-t)$ is bounded on $[0,T]$).

\begin{theorem}[\textbf{Verification Theorem}]
\label{verif_full_info}
Let $\Theta$ be as in Theorem \ref{Theta_constrained_inventory} for $\Nmax<+\infty$ and as in Proposition \ref{Theta_unconstrained_inventory} for $\Nmax=+\infty$ and $\gamma=\zeta=0\equiv\ell$. Then the value function in (\ref{value_function_full_info}) is 
$$
V(t,s,x,n,i)=U_\gamma(x+sn+\Theta(t,n,i)).
$$
Furthermore, given $C>0$ such that $|\Theta(t,n\mp 1,i)-\Theta(t,n,i)|\leq C$ for all $(t,n,i)\in[0,T]\times\mathcal I$ with $-\Nmax\leq n\mp 1\leq\Nmax$, there exist Borel measurable functions $\overline{\delpm_1},\dots,\overline{\delpm_k}:[-C,C]\to\overline{(\delmin,\delmax)}$ such that $\overline{\delpm_i}(d)\in\argmax_{\delta\in\overline{(\delmin,\delmax)}}\hpm_i(\delta,d)$ for all $d\in[-C,C]$, $1\leq i\leq k$; and for any such functions, the strategy $(\overline{\delp_u},\overline{\delm_u})$, with
$$
\overline{\delpm_u}\defeq\overline\delpm_{\Yt_{u^-}}\Big(\Theta(u,\Nt_{u^-}\mp 1,\Yt_{u^-})-\Theta(u,\Nt_{u^-},\Yt_{u^-})\Big),
$$
is optimal. 
\end{theorem}
\begin{notation*}
Note that we make a slight abuse of notation, writing $\overline{\delpm_i}$ for the Borel functions in the theorem and $(\overline{\delpm_u})$ for the spread processes. 
\end{notation*}

\begin{proof}
Let $1\leq i\leq k$. We check first that we can choose a maximizer of $\hpm_i(\cdot,d)$ in a measurable way with respect to $d$. Lemma \ref{lipschitz_hamiltonians} tells us that there exists $[a_i,b_i]\subseteq\overline{(\delmin,\delmax)}$ such that $\hpm_i(\delta,d)$ attains its maximum in $[a_i,b_i]$ for all $d\in[-C,C]$. 
Due to Assumptions \ref{assumptions_intensities} \textit{(\ref{assumptions_intensities_Lambda})}, $\hpm_i$ is continuous and, in particular, a Carathéodory function \cite[Def.4.50]{AB}. Thus, the Measurable Maximum Theorem \cite[Thm.18.19]{AB} guarantees the existence of a Borel selector of the $\argmax$, $\overline{\delpm_i}:[-C,C]\to[a_i,b_i]\subseteq\overline{(\delmin,\delmax)}$. (Note that the weak measurability assumption is trivially verified in this case.)

From here onwards, let $\overline{\delpm_1},\dots,\overline{\delpm_k}:[-C,C]\to\overline{(\delmin,\delmax)}$ be some Borel selectors as above. By Lemma \ref{lipschitz_hamiltonians} again, these functions must be bounded from below, and the spread processes $(\overline{\delpm_u)}$ defined as in the theorem are clearly admissible.

We define now $\tilde V(t,s,x,n,i)\defeq U_\gamma(x+ns+\Theta(t,n,i))$ and we want to show that $\tilde V=V$.
Let us fix the initial time and values, $t\in[0,T],\ s,x\in\mathbb R\mbox{ and }(n,i)\in\mathcal I$, and consider an arbitrary strategy $\alpha=(\delm_u,\delp_u)\in\mathcal U_t^2$. For shortness, we omit these initial conditions and strategy from the notation of the processes and the expectation. 
We denote by $\mathcal S\defeq \{j-l:\ 1\leq l,j\leq k,\ l\neq j\}$, the set of jump heights of $Y$, and
$$
R_{t,v}\defeq -\frac{1}{2}\sigma^2\zeta \int_t^v N_u^2 du,\qquad Z_{t,v}\defeq \exp\big(-\gamma R_{t,v}\big).
$$
By the identity $U_\gamma(a+b)=U_\gamma(a)e^{-\gamma b}+U_\gamma(b)$, we can rewrite the utility of the MM's penalized P\&L as 
\begin{equation}
\begin{split}
\label{eq1}
U_\gamma\big(P_{t,T}\big) & = U_\gamma\left(X_T+S_TN_T-\ell(N_T)\right) Z_{t,T} + U_\gamma\big( R_{t,T}\big)\\
& = \tilde V\big(T,S_T,X_T, N_T,Y_T\big) Z_{t,T}+U_\gamma\big( R_{t,T}\big).
\end{split}
\end{equation}
Using integration by parts and Ito's Lemma (recalling (\ref{no_common_jumps})), we re-express the last two terms as 
\begin{align*}
U_\gamma\big( R_{t,T}\big) = -\frac{1}{2}\sigma^2\zeta\int_t^T Z_{t,u}N_u^2du
\end{align*}
and
\begin{align*}
& \qquad\tilde V\big(T,S_T,X_T, N_T,Y_T\big) Z_{t,T}\\
& =\tilde V(t,s,x,n,i)+\int_t^T \tilde V (u,S_u,X_u, N_u,Y_u )d Z_{t,u} + \int_t^T Z_{t,u}d\tilde V (u,S_u,X_u, N_u,Y_u )\\
& = \tilde V(t,s,x,n,i)+\frac{1}{2}\sigma^2\zeta\gamma\int_t^T  V (u,S_u,X_u, N_u,Y_u ) Z_{t,u} N_u^2du+\int_t^T  Z_{t,u}\tilde V_t (u, S_u, X_{u^-}, N_{u^-},Y_{u^-} ) du\\
& + \mu \int_t^T  Z_{t,u}\tilde V_s (u, S_u, X_{u^-}, N_{u^-},Y_{u^-} )du+ \sigma \int_t^T  Z_{t,u}\tilde V_s (u, S_u, X_{u^-}, N_{u^-},Y_{u^-} ) dW_u\\
& + \frac{1}{2}\sigma^2\int_t^T  Z_{t,u}\tilde V_{ss} (u, S_u, X_{u^-}, N_{u^-},Y_{u^-} ) du\\
& +\int_t^T  Z_{t,u}\big( \tilde V(u, S_u, X_{u^-} - (S_u - \delm_u), N_{u^-}+1,Y_{u^-}) - \tilde V(u, S_u, X_{u^-}, N_{u^-},Y_{u^-})\big) d\Nm_u\\
& + \int_t^T  Z_{t,u}\big( \tilde V(u, S_u, X_{u^-} + S_u + \delp_u, N_{u^-}-1,Y_{u^-}) - \tilde V(u, S_u, X_{u^-}, N_{u^-},Y_{u^-})\big) d\Np_u \\
& + \int_t^T Z_{t,u}\int_{\mathcal S}\big( \tilde V (u, S_u, X_{u^-}, N_{u^-},Y_{u^-}+h ) - \tilde V (u, S_u, X_{u^-}, N_{u^-},Y_{u^-} ) \big)\mu^Y(dt,dh),
\end{align*}
where $\mu^Y$ is the jump measure of $Y$. Substituting in (\ref{eq1}),

\begin{equation}
\label{eq2}
\begin{split}
& U_\gamma\big(P_{t,T}\big) = \tilde V(t,s,x,n,i) + \sigma \int_t^T  Z_{t,u}\tilde V_s (u, S_u, X_{u^-}, N_{u^-},Y_{u^-} ) dW_u\\
&+\int_t^T  Z_{t,u} \Big\{\tilde V_t (u, S_u, X_{u^-}, N_{u^-},Y_{u^-} )du+\mu \tilde V_s (u, S_u, X_{u^-}, N_{u^-},Y_{u^-} )du\\
&+\frac{1}{2}\sigma^2\tilde V_{ss} (u, S_u, X_{u^-}, N_{u^-},Y_{u^-} )du + \frac{1}{2}\sigma^2\zeta N_u^2
\big(\gamma V(u,S_u,X_u, N_u,Y_u )-1\big)du\\
& +\int_{\mathcal S}\big( \tilde V(u, S_u, X_{u^-}, N_{u^-},Y_{u^-}+h) - \tilde V(u, S_u, X_{u^-}, N_{u^-},Y_{u^-}) \big)\mu^Y(dt,dh)\\
& +\big( \tilde V(u, S_u, X_{u^-} - (S_u - \delm_u), N_{u^-}+1,Y_{u^-}) - \tilde V(u, S_u, X_{u^-}, N_{u^-},Y_{u^-})\big) d\Nm_u \\
& + \big( \tilde V(u, S_u, X_{u^-} + S_u + \delp_u, N_{u^-}-1,Y_{u^-}) - \tilde V(u, S_u, X_{u^-}, N_{u^-},Y_{u^-})\big) d\Np_u\Big\}.
\end{split}
\end{equation}
Next, we want to verify that the process $\Big(\int_t^v  Z_{t,u}\tilde V_s (u, S_u, X_{u^-}, N_{u^-},Y_{u^-} ) dW_u\Big)_{t\leq v\leq T}$ is a zero mean martingale. Since $Z_{t,u}$ is bounded (either $\Nmax<\infty$ or $\gamma=0$) it suffices to check
\begin{equation}
\label{eq3}
\int_t^T \mathbb E\left[ \tilde V_s (u, S_u, X_u, N_u,Y_u )^2\right] du<+\infty.
\end{equation}
If $\gamma=0$, 
$$
\tilde V_s (u, S_u, X_u, N_u,Y_u )^2={N_u}^2
$$
and (\ref{eq3}) is a consequence of (\ref{bounded_moments}). 
\bigskip

\noindent If $\gamma\neq 0$ (hence $\Nmax<\infty$), let $C$ be a lower bound for $(\delm_u)$ and $(\delp_u)$. Integration by parts and Hölder inequality yield
\begin{align*}
&\mathbb E\left[\tilde V_s (u, S_u, X_u, N_u,Y_u )^2\right] =\gamma^2\mathbb E\left[\exp\big(-2\gamma\big(X_u+S_uN_u + \theta(u,N_u,Y_u)\big)\big)N_u^2 \right]\\
& =\gamma^2\exp(-2\gamma(x+sn))\\
&\times\mathbb E\left[\exp\Big(-2\gamma\int_t^u\delm_wd\Nm_w+ \delp_wd\Np_w +\mu N_wdw+\sigma N_wdW_w+\theta(w,N_w,Y_w)\Big) N_u^2 \right]\\
& \leq \gamma^2{\Nmax}^2\exp\Big(2\gamma\big(x+sn+\mu\Nmax (T-t)+\|\theta\|_{L^\infty([t,T]\times\mathcal I)}+2\sigma^2\gamma{\Nmax}^2(T-t)\big)\Big)\\
&\times{\mathbb E\left[\exp\Big(-4\gamma  C  \big(\Nm_T+\Np_T-\Nm_{t^-}-\Np_{t^-}\big)\Big)\right]}^{\frac{1}{2}}{\mathbb E\left[\exp\Big(-4\gamma\sigma\int_t^T N_u dW_u - 8\gamma^2\sigma^2\int_t^T N_u^2 du\Big)\right]}^{\frac{1}{2}}
\end{align*}
and (\ref{eq3}) is again consequence of (\ref{bounded_moments}) and Novikov's condition (trivially satisfied).
\bigskip

\noindent In the same way, recalling that $(Z_{t,u})$ is bounded, $(\delpm_u)$ is bounded from below and that $Q$ is bounded by continuity, one can check that 
$$
\int_t^T \mathbb E\Big[ Z_{t,u}\big| \tilde V(u, S_u, X_{u^-} \pm (S_u \pm \delpm_u), N_{u^-}\mp 1,Y_{u^-}) - \tilde V(u, S_u, X_{u^-}, N_{u^-},Y_{u^-})\big|\lpm_u\Big] du<+\infty
$$
and 
$$
\int_t^T \mathbb E\Big[Z_{t,u}\sum_{j\neq Y_{u^-}}q_t^{Y_{u^-},j}\big| \tilde V(u, S_u, X_{u^-}, N_{u^-},j) - \tilde V(u, S_u, X_{u^-}, N_{u^-},Y_{u^-}) \big|\Big] du<+\infty.
$$

\noindent Taking expectation in (\ref{eq2}), the Brownian term vanishes and integration with respect to $d\Nm,d\Np$ and $d\mu^Y$ is replaced by integration with respect to their dual predictable projections (see, e.g., \cite[p.27 T8 and p.235 C4]{B}).
That is,
 
\begin{equation*}
\label{eq4}
\begin{split}
& \mathbb E\left[U_\gamma\big(P_{t,T}\big)\right] = \tilde V(t,s,x,n,i)+ \mathbb \int_t^T \mathbb E\Big[ Z_{t,u} \Big\{\tilde V_t(u, S_u, X_{u^-}, N_{u^-},Y_{u^-} )\\
& +\mu \tilde V_s (u, S_u, X_{u^-}, N_{u^-},Y_{u^-} ) + \frac{1}{2}\sigma^2\tilde V_{ss} (u, S_u, X_{u^-}, N_{u^-},Y_{u^-} )\\
& + \frac{1}{2}\sigma^2\zeta N_u^2\big(\gamma \tilde V(u,S_u,X_u, N_u,Y_u)-1\big)+\sum_{j=1}^k q^{Y_{u^-}j}_t \tilde V(u, S_u, X_{u^-}, N_{u^-},j )\\
& +\mathbbm 1^-_u\Lm_{Y_{u^-}}(\delm_u)\big( \tilde V(u, S_u, X_{u^-} - (S_u - \delm_u), N_{u^-}+1,Y_{u^-}) - \tilde V(u, S_u, X_{u^-}, N_{u^-},Y_{u^-})\big)\\
& +\mathbbm 1^+_u\Lp_{Y_{u^-}}(\delp_u)\big( \tilde V(u, S_u, X_{u^-} + S_u + \delp_u, N_{u^-}-1,Y_{u^-}) - \tilde V(u, S_u, X_{u^-}, N_{u^-},Y_{u^-})\big) \Big\}\Big] du\\
& \leq  \tilde V(t,s,x,n,i),
\end{split}
\end{equation*}
as $\tilde V$ solves (\ref{HJB_V_full_info}), with the equality attained for $(\delm_u)=(\overline{\delm_u}) \mbox{ and } (\delp_u)=(\overline{\delp_u})$  by definition.
We conclude that $V=\tilde V$ and that the pair $(\overline{\delm_u}),(\overline{\delp_u})$ is optimal.
\end{proof}

\begin{remark*}
For $\Nmax<\infty$, since the MM will not buy (resp. sell) whenever $(N_{u^-})$ hits $\Nmax$ (resp. $-\Nmax$), the value of $(\overline{\delm_u})$ (resp. $(\overline{\delp_u})$) at these stopping times is essentially irrelevant. From a strict mathematical perspective, the only constraint is that whichever the value we choose, the process needs to remain admissible.  
\end{remark*}

\subsection{Computing the optimal spreads: some particular cases}
\label{s:computing_spreads}
We have shown in Theorem \ref{verif_full_info} that the optimal spreads for the full information problem (\ref{value_function_full_info}) can be computed in feedback form in terms of $(t,n,i)=(t,N_{t^-},Y_{t^-})$ at each time $t\in[0,T]$. Practically, this means finding $\Theta$ (typically, numerically) that solves the terminal condition system of ODEs (\ref{HJB_Theta_full_info}), and finding the spreads by maximization:
\begin{equation} 
\label{eq1_argmax}
\overline{\delpm}(t,n,i)\in\argmax_{\delta\in\overline{(\delmin,\delmax)}}\Lipm(\delta)U_\gamma\left(\delta +\Theta(t,n\mp 1,i)-\Theta(t,n,i)\right).
\end{equation}
(See the proof of Lemma \ref{lipschitz_hamiltonians} to see how to reduce the maximization to a compact domain in the cases of $\delmin=-\infty$ or $\delmax=+\infty$.)
The functions in (\ref{eq1_argmax}) may admit multiple maximizers in general. In \cite{BL, G, GL} stronger assumptions are imposed on the orders intensities which guarantee, in particular, the uniqueness of the maximizers.
We now extend these assumptions to our context and give the corresponding characterization of the spreads. Henceforth, Assumption \ref{assumptions_intensities} \textit{(\ref{assumptions_intensities_Lambda})} is replaced by the following:

\begin{assumptions}
\label{assumptions_intensities_strong}
$\Lipm\in C^2\big(\overline{(\delmin,\delmax)}\big), \quad {\Lipm}'<0\mbox{ and } \Lipm{\Lipm}''<c({\Lipm}')^2$ for some $0<c<2$, for all $1\leq i\leq k$.\footnote{Lateral derivatives are considered on the domain border. The strict inequality $c<2$ is needed when $\gamma=0$ and spreads are unconstrained.}
\end{assumptions}

\begin{remark*}
The intensity functions of Examples \ref{ex} \ref{ex1}, \ref{ex2} and \ref{ex3} all verify these stronger assumptions, while Example \ref{ex4} only does so for $\Dpm_i>1$. As a matter of fact, by solving the differential inequality in Assumptions \ref{assumptions_intensities_strong}, one sees that any $\Lipm$ within this new framework has a strict upper bound of the form of Example \ref{ex4}. In particular, for $\delmax=+\infty$, $\lim_{\delta\to+\infty}\delta\Lipm(\delta)=0$ is necessarily satisfied. 
\end{remark*}

Under Assumptions \ref{assumptions_intensities_strong}, the maximization of (\ref{eq1_argmax}) (for fixed $(t,n,i)$) can be replaced essentially by solving a contractive fixed point equation and flooring and capping the results at $\delmin$ and $\delmax$ respectively. To make this precise, even when the intensity functions are not defined beyond $\overline{(\delmin,\delmax)}$, let us set	 
$$
\Dpm_{i*}\defeq 
\begin{cases}
-U_\gamma^{-1}\left(\frac{\Lipm(\delmin)}{{\Lipm}'(\delmin)}\right) - \delmin &\mbox{ if }\delmin>-\infty\\
+\infty &\mbox{ if }\delmin=-\infty,
\end{cases}
\qquad 
{\Dpm_i}^*\defeq 
\begin{cases}
-U_\gamma^{-1}\left(\frac{\Lipm(\delmax)}{{\Lipm}'(\delmax)}\right) - \delmax &\mbox{ if }\delmax<+\infty\\
-\infty &\mbox{ if }\delmax=+\infty. 
\end{cases}
$$
The following result shows how, up to constraints, the optimal spreads are given by a first term that maximizes the ``expected" utility of the MM's instantaneous margin (see Remark \ref{expected_margin}), plus an additional risk adjustment taking into account the inventory held and the prospect of the market shifting. Note how the first term depends on ${\Lipm}'/\Lipm$, the percentage sensitivity of the liquidity to spread changes.

\begin{proposition}
\label{prop_flooring_capping}
Under Assumptions \ref{assumptions_intensities_strong}, the functions $\overline{\delpm_1},\dots,\overline{\delpm_k}$ of Theorem \ref{verif_full_info} are uniquely characterized by
\begin{equation}
\label{flooring_capping}
\overline{\delpm_i}(d)= \delmax\mbox{ if }\ d<{\Dpm_i}^*,\qquad\overline{\delpm_i}(d)= \delmin\mbox{ if }\ d>\Dpm_{i*}
\end{equation}
and if $d\in\overline{({\Dpm_i}^*,\Dpm_{i*})}$, then $\overline{\delpm_i}(d)$ is the unique solution of the fixed point equation
\begin{equation}
\label{eq2_fixed_point}
\overline{\delpm_i}(d)= -U_\gamma^{-1}\left(\frac{\Lipm(\overline{\delpm_i}(d))}{{\Lipm}'(\overline{\delpm_i}(d))}\right) - d.
\end{equation}
Additionally, the Hamiltonians of equation (\ref{HJB_Theta_full_info}) can be expressed as
\begin{equation*}
\label{H_constrained_explicit}
\Hpm_i(d)=
\begin{cases}
\hpm_i(\delmax,d)\qquad\qquad\qquad\qquad\mbox{ if }d\leq{\Dpm_i}^*\\
-\frac{{\Lipm}^2(\widehat\delpm(d))}{{\Lipm}'(\widehat\delpm(d))-\gamma\Lipm(\widehat\delpm(d))}\qquad\,\mbox{ if }{\Dpm_i}^*< d<\Dpm_{i*}\\
\hpm_i(\delmin,d)\qquad\qquad\qquad\qquad\mbox{ if }d\geq\Dpm_{i*}.
\end{cases}
\end{equation*}
\end{proposition}

\begin{proof}
For each $1\leq i\leq k$, define 
$$
\fipm(\delta,d)=\delta+U_\gamma^{-1}\left(\frac{\Lipm(\delta)}{{\Lipm}'(\delta)}\right)+d,\quad\mbox{with }\delta\in\overline{(\delmin,\delmax)},\ d\in\mathbb R
.$$
By Assumptions \ref{assumptions_intensities_strong} and straightforward computations (see computations, e.g., in \cite[Lemma 3.1]{G}), one verifies that $\sgn\big(\frac{\partial\hpm_i}{\partial\delta}\big)=-\sgn(\fipm)$ and $\fipm$ is strictly increasing on each variable.\footnote{$\sgn$ denotes the sign function, with $\sgn(0)\defeq 0$.} It follows that for any $d<{\Dpm_i}^*$ and $\delta<\delmax$ (resp. $d>\Dpm_{i*}$ and $\delta>\delmin$), $\fipm(\delta,d)<\fipm(\delmax,{\Dpm_i}^*)=0$ (resp. $\fipm(\delta,d)>\fipm(\delmin,\Dpm_{i*})=0$), which proves (\ref{flooring_capping}). 

In the same way, if $d\in({\Dpm_i}^*,\Dpm_{i*})$ and $-\infty<\delmin<\delmax<+\infty$, then $\fipm(\delmin,d)<0$ and $\fipm(\delmax,d)>0$. Hence, the continuity of $\fipm(\cdot,d)$ implies (\ref{eq2_fixed_point}) (the uniqueness of the solution being due to the strict monotonicity of $\fipm(\cdot,d)$). The cases of unconstrained spreads are proved as in \cite{G}. 

Lastly, the cases $d={\Dpm_i}^*>-\infty$ and $d=\Dpm_{i*}<+\infty$ follow in the same manner; and the new expressions for the Hamiltonians are immediate, putting $\Hpm_i(d) = \hpm_i\big(\overline{\delpm_i}(d),d\big)$.
\end{proof}

As done in \cite[Lemma 3.1]{G}, (\ref{eq2_fixed_point}) can be replaced by the explicit formula 
$$
\overline{\delpm_i}(d) = (\Lipm)^{-1}\left(\gamma \Hpm_i(d) -{\Hpm_i}'(d) \right), 
$$
but the computation of $(\Lipm)^{-1}$, $\Hpm_i$ and ${\Hpm_i}'$ still has to be carried out numerically, in general. We state now a few simplifications that arise in some particular cases, i.e., when the intensities are exponential or when $\Nmax=\infty$, $\gamma=\zeta=0\equiv\ell$. These are all derived by straightforward substitution. We refer to \cite[Sect.4]{G} for some asymptotic approximations when $t<<T$ in the one regime case, with unconstrained spreads and $\Nmax<\infty$. 

\begin{corollary}
\label{coro_spreads_vanilla_model}
If $\Nmax=+\infty$ and $\gamma=\zeta=0\equiv\ell$, then 
\begin{equation}
\overline{\delpm}(t,i)= \delmax\mbox{ if }\ \mp\mu(T-t)<{\Dpm_i}^*,\qquad\overline{\delpm_i}(d)= \delmin\mbox{ if }\ \mp\mu(T-t)>\Dpm_{i*},
\end{equation}
and if $\mp\mu(T-t)\in\overline{({\Dpm_i}^*,\Dpm_{i*})}$, then $\overline{\delpm}(t,i)$ is the unique solution of the fixed point equation
\begin{equation}
\label{eq3_fixed_point}
\overline{\delpm}(t,i) = -\frac{\Lipm(\overline{\delpm}(t,i))}{{\Lipm}'(\overline{\delpm}(t,i))} \pm \mu (T-t).
\end{equation}
If additionally $\Lipm(\delta)=\apm_i e^{-\bpm_i\delta}$, with $\apm_i,\bpm_i>0$ for each $1\leq i\leq k$, then 
$$ \overline{\delpm}(t,i) = \delmin\vee\left(\frac{1}{\bpm_i}\pm\mu (T-t)\right)\wedge\delmax.$$
\end{corollary}
In other words, if $\Nmax=\infty$ and $\gamma=\zeta=0\equiv\ell$, then solving (\ref{HJB_Theta_full_info}) is no longer necessary and one only needs to solve the fixed point equations (\ref{eq3_fixed_point}). Further, the MM's spreads do not depend on $n$. This is to be expected, since in this case she is neutral to all types of inventory risks and the terminal execution cost is neglected. However, she does need to re-adjust for changes in the regime as these impact on her probability of getting orders. If additionally, the orders intensities are exponential, then the spreads are computed straightforwardly avoiding the need of any numerical scheme. Note also how this simple model makes evident a second component of the optimal spreads through equation (\ref{eq3_fixed_point}): a drift adjustment by which the MM takes into account the overall tendency of the asset's price.  

Now we look at the results obtained for exponential intensities in general, and in particular for $\Nmax<+\infty$. Equation (\ref{eq2_fixed_point}) becomes an explicit formula in this case, as the percentage liquidity sensitivities,  ${\Lipm}'/\Lipm$, are constant.
\begin{corollary}
If $\Lipm(\delta)=\apm_i e^{-\bpm_i\delta}$, with $\apm_i,\bpm_i>0$ for each $1\leq i\leq k$, then
$$\overline{\delpm}(t,n,i) =\delmin\vee \left( -U_\gamma^{-1}\left(-1/\bpm_i\right)-\big(\Theta(t,n\mp 1,i)-\Theta(t,n,i)\big)\right)\wedge\delmax.$$
\end{corollary}
\noindent This means that for exponential intensities in general, it is no longer necessary to solve any fixed point equations but only the system of ODEs (\ref{HJB_Theta_full_info}). The latter still needs to be solved numerically in general. (Alternatively, see \cite{G} for some asymptotic approximations.) 

\begin{remark*}
Apart from the case $\Nmax=+\infty,\ \gamma=\zeta=0\equiv\ell$, a noteworthy scenario in which (\ref{HJB_Theta_full_info}) can be further simplified to a linear system of ODEs (with constant coefficients) is when $k=1,\ \Nmax<+\infty,\ \delmin=-\infty,\ \delmax=+\infty$ and $\Lpm(\delta)=a e^{-b\delta}$, with $a,b>0$. The reduction is achieved via the transformation 
$
\Theta(t,n)=\frac{1}{b}\log\Psi(t,n).
$
See \cite{GLFT} for more details.
\end{remark*}

\section{Numerical analysis}
\label{s:numerics}
\setcounter{equation}{0}
\setcounter{subsection}{0}

In this section we present our numerical results, focusing on the difference in optimal behaviours between the partial and full information frameworks, and the intuition behind the filter. To exemplify our findings in a concrete manner, while keeping the presentation as simple as possible, we will assume throughout this section that: 
\begin{itemize}
\item The MM has risk aversion parameter $\gamma=0$.
\item There are only two possible states: state 1 represents a ``bad" regime with low liquidity taken by clients, and state 2 a ``good" regime with high liquidity.
\item The transition rate matrix $Q$ is constant.
\item The intensities are symmetric, proportional and exponential, i.e., $\Lpm_1(\delta)=\Lambda_1(\delta)=ae^{-b\delta}$ and $\Lpm_2(\delta)=\Lambda_2(\delta)=m\Lambda_1(\delta),\mbox{ with }a,b>0,\ m>1$.
\end{itemize}
The latter assumption allows us to perform the optimizations in equation (\ref{HJB_Theta}) analytically. Practically, proportional intensities mean that while there is more active trading on the good regime, the way in which the clients react to movements in the spreads remains unaffected. As in \cite{CJ}, we will allow for the three type of penalties to manage inventory risk: constraints on the maximum long and short positions, accumulated inventory penalty and a quadratic (possibly negligible) terminal penalty (or cost) for the MM. That is, 
\begin{itemize}
\item $\Nmax<+\infty,\ \zeta\geq 0\mbox{ and }\ell(n)=cn^2$, for some $c\geq 0$.
\end{itemize}

Let us write $\pi\defeq \pi_1$ for the conditional probability of being in the bad regime given the observable information. Note that $\pi$ is a scalar in this section, since $\pi_2=1-\pi_1$ neglects the need of the additional variable. The PIDE (\ref{HJB_Theta}) at the point $(t,n,\pi)$ reads:
\begin{equation}
\label{HJB_numerics}
0 = \theta_t + \mu n - \frac{1}{2}\sigma^2 n^2 \zeta+ 2\hat q(\pi)\theta_\pi +\frac{a}{b}\hat m(\pi)\left(\mathbbm 1_{\{n<\Nmax\}} e^{-b\widehat\delm} + \mathbbm 1_{\{n>-\Nmax\}} e^{-b\widehat\delp} \right),
\end{equation}
with terminal condition $\theta(T,n,\pi) = -cn^2$ and partial information optimal spreads given by
\begin{equation}
\label{partial_info_optimal_spreads}
\widehat\delpm(t,n,\pi)=\frac{1}{b} - 2\frac{w(\pi)}{\hat m(\pi)}\theta_\pi(t,n,\pi) - \big(\theta(t,n\mp 1,\pi/\hat m(\pi))-\theta(t,n,\pi)\big),\mbox{ where:}
\end{equation}
\begin{enumerate}[(i)]
\item $\hat q(\pi) = q^{11}\pi + q^{21}(1-\pi)$ is the observable transition rate to the bad regime.
\item $\hat m(\pi) = \pi + (1-\pi)m$ is the observable intensity increase from the bad regime (as a ratio).
\item $w(\pi)=(m-1)\pi(1-\pi)$ is the observable variance, of the square root, of the percentage intensity increase from the bad regime; i.e., a measure of observable order flow volatility.
\end{enumerate}
The previous equations are valid for $\delmin\ll 0\ll\delmax$, or more precisely, for $\delmin,\delmax$ such that $\delmin\leq\widehat\delpm\leq\delmax$ holds true over the whole domain. Otherwise, one needs to floor and cap the optimal spreads and change the Hamiltonians accordingly, as done in Proposition \ref{prop_flooring_capping}. 

A finite differences scheme of simple implementation to solve (\ref{HJB_numerics}) consists of reversing time and using an explicit upwind Euler scheme over a uniform grid for $(t,\pi)$ and $n=-\Nmax,-\Nmax+1,\dots,\Nmax$. The terms of the form $\theta(t,n\mp 1,\pi/\hat m(\pi))$, where $\pi/\hat m(\pi)$ will typically fall outside of the grid, can be approximated by linear interpolation as in \cite{DFo}. The limiting equations can be used for $\pi\to 0^+$ and $\pi\to 1^-$ provided that $-q^{11}>0$, which we assume henceforth. Such a scheme can be shown to be consistent, stable and monotone under an appropriate CFL condition.
In light of equation (\ref{HJB_Theta}) satisfying a comparison principle (see \cite[Thm.5.3]{CEFS} and why it holds for our model in the proof of Theorem \ref{main_theorem_2}), we know the scheme converges to the unique continuous viscosity solution \cite{BS}, and we can recover the expected penalized P\&L of the MM (Theorems \ref{main_theorem_1} and \ref{main_theorem_2}). The optimal strategy to be followed by the MM is then given in feedback form by $\big(\widehat\delm(t,N_{t^-},\Pi_{t^-}),\widehat\delp(t,N_{t^-},\Pi_{t^-})\big)$.\footnote{As previously mentioned, these are technically only candidates for optimal (or $\epsilon$-optimal) strategies. We do not rigorously prove their optimality character here, but merely note that well known results of discrete time dynamic programming, together with the convergence of the discrete solutions of (\ref{HJB_numerics}) to the analytical one, suggest that they can be safely used as such.}

We will focus our attention on the optimal ask spread, with analogous observations holding for the bid spread. The parameter values in Table \ref{table_parameters} were used for all the experiments presented in this section. For those present in the classical one regime models, we have chosen values used in previous works (in the $\gamma=0$ case \cite{CJ}) to make the comparison clearer. In particular, the value of $c$ will be taken as either $0$ or $0.01$, to be further specified in each experiment. The time horizon will always be the one displayed on the corresponding axis. Note that we work in a symmetric market, which justifies analyzing one side only.

\begin{table}[h!]
\centering
 \begin{tabular}{||c |c |c |c |c |c |c |c |c |c |c ||} 
 \hline
 Parameter & $\mu$ & $\sigma$ & $\zeta$ & $a$ & $b$ & $\Nmax$ & m & $-q^{11}=q^{21}$ & $-\delmin=\delmax$ \\ [0.5ex] 
 \hline
 Value & 0 & 0.1 & 0.1 & 2 & 25 & 3 & 5 & 5 & 10 \\
 \hline
\end{tabular}
\caption{Standing parameter values for numerical tests presented.}
\label{table_parameters}
\end{table}
\subsection{Comparing full and partial information optimal strategies}
Under the standing assumptions of this section, the full information equation (\ref{HJB_Theta_full_info}), for a function $\tilde\theta$, becomes:
\begin{equation}
\label{HJB_full_info_numerics}
0 = \tilde\theta_t + \mu n - \frac{1}{2}\sigma^2 n^2 \zeta+ \tilde q_i \left(\tilde\theta(t,n,1)-\tilde\theta(t,n,2)\right) +
\frac{a}{b} \tilde m_i \left(\mathbbm 1_{\{n<\Nmax\}} e^{-b\widetilde\delm} + \mathbbm 1_{\{n>-\Nmax\}} e^{-b\widetilde\delp} \right),
\end{equation}
with terminal condition $\tilde\theta(T,n,i) = -cn^2$ and full information optimal spreads given by
\begin{equation}
\label{full_info_optimal_spreads}
\widetilde\delpm(t,n,i)=\frac{1}{b} - \big(\tilde\theta(t,n\mp 1,i)-\tilde\theta(t,n,i)\big),\mbox{ where:}
\end{equation}
\begin{enumerate}[(i)]
\item $\tilde q_i =q^{11}\mathbbm 1_{\{1\}}(i)+q^{21}\mathbbm 1_{\{2\}}(i)$ is the effective transition rate to the bad regime.
\item $\tilde m_i = \mathbbm 1_{\{1\}}(i) + \mathbbm 1_{\{2\}}(i) m$ is the effective intensity increase from the bad regime (as a ratio).
\end{enumerate}
\noindent The previous equations are valid for $\delmin\ll 0\ll\delmax$, as in the partial information case.

Although similar, equation (\ref{HJB_full_info_numerics}) for the bad regime $i=1$ (resp. good regime $i=2$) is not the limiting equation of (\ref{HJB_numerics}) for $\pi\to 1^-$ (resp. $\pi\to 0^+$). Indeed, a MM with full information can expect to make a larger profit. Thus, in general, $\tilde\theta>\theta$ even in these extreme cases.\footnote{Our numerical findings were indeed consistent with this intuitive statement.} However, the corresponding optimal strategies do (at least approximately) agree, as Figures \ref{delp_full_vs_partial_running_penalty} and \ref{delp_full_vs_partial_all_penalties} show.

\begin{figure}[htb]
\hspace*{-.38cm}
	\includegraphics[scale=.33]{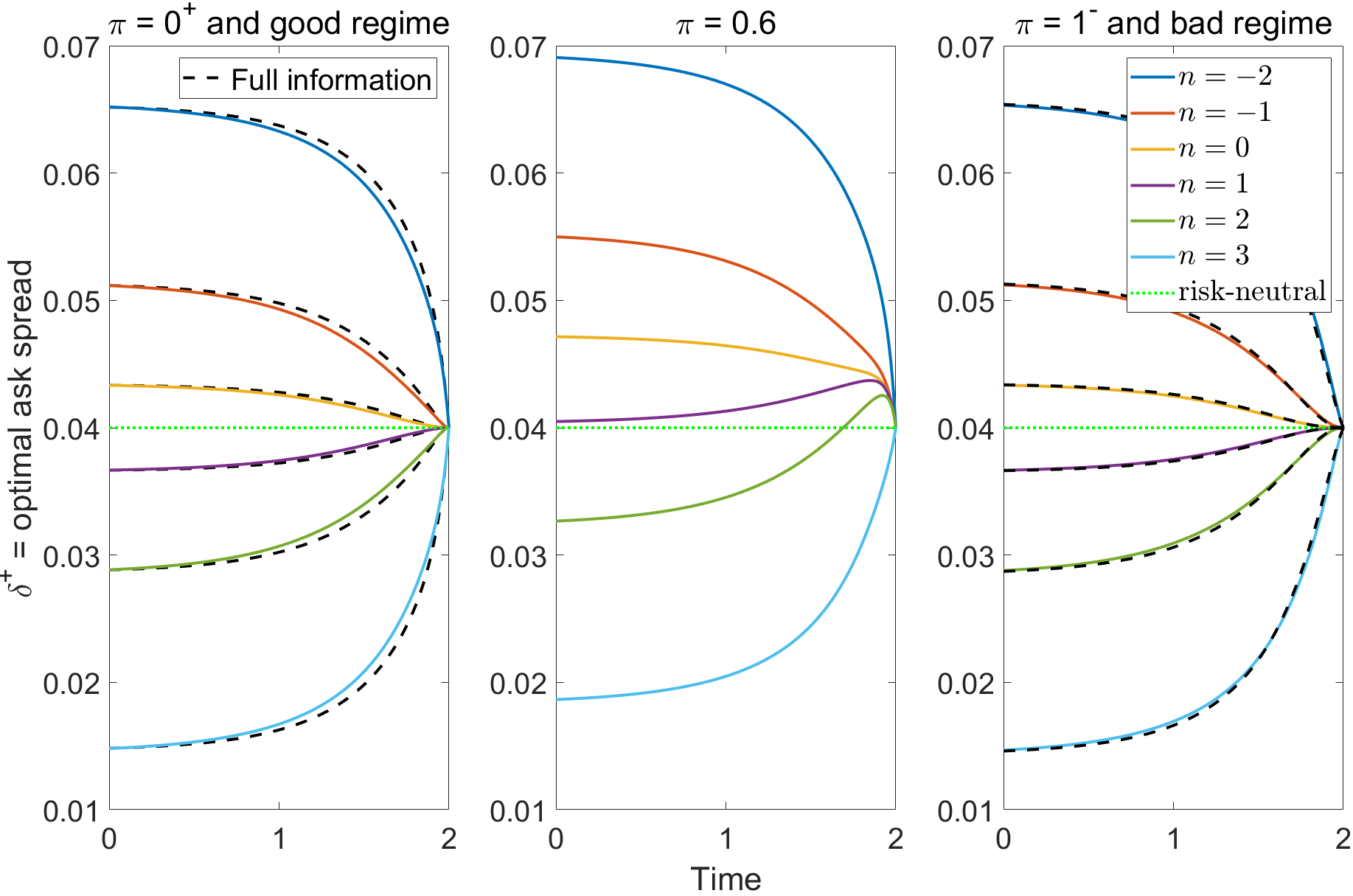}
	\caption{Optimal ask spread as a function of time, under full and partial information (dashed and solid lines resp.), and for different inventory and filter values. Inventory levels increase from top to bottom in all cases. Terminal execution cost is neglected ($c=0$).}
	\label{delp_full_vs_partial_running_penalty}
\end{figure}

Figure \ref{delp_full_vs_partial_running_penalty} shows the optimal ask spread under partial information as a function of time, for different inventory levels (solid lines). It also displays how it changes for different filter values ($\pi=0,\ 0.6\mbox{ and }1$ from left to right) and how it compares to the optimal ask spread under full information (dashed lines) in the good regime (left) and bad regime (right). The inventory levels increase from top to bottom in all cases, and the terminal execution cost is neglected ($c=0$). We have chosen to display $\pi=0.6$ due to the effect of the filter being more pronounced around this value; it results in an optimal spread which is, for $t\ll T$, between $6\%$ and $30\%$ higher (depending on the inventory level) than the corresponding ones under full information (a sizeable difference in practice). We will come back to this in Section \ref{section_simulations}. All other values show similar intermediate behaviours. Recall that when the inventory reaches the minimum $n=-\Nmax=-3$, the MM will not sell any more (either abstaining from quoting or giving a ``stub" quote) until her inventory increases again, which is why there is no spread plotted for this position. We see that some of the features already present in one regime models are preserved for two regimes, both under full and partial information; namely: 

\begin{itemize}

\item For $t\ll T$ the spread does not depend on time (approximately) and its asymptotic value becomes higher as the inventory decreases to the lower constraint $n=-\Nmax$. Indeed, for flat or short positions, any additional ask order will increase inventory risk, moving it closer to the minimum allowed and raising the accumulated inventory penalty. The MM manages this risk by increasing the spread, thus demanding a higher instantaneous profit and reducing the probability of getting executed. Similarly for long positions, the higher the exposure the lower the spread, as the MM seeks to unwind her inventory.

\item In the case of negligible terminal cost, the spread converges to a terminal value independent of the inventory. This is the optimal spread of a fully risk-neutral MM with negligible costs (Corollary \ref{coro_spreads_vanilla_model}), who only maximizes ``expected instantaneous margin". (This value is given by $\widehat\delp(T,\cdot,\cdot) = 1/b$, the reciprocal percentage sensitivity of the liquidity to spread changes; cf. Proposition \ref{prop_flooring_capping}). The reason being that as the time horizon approaches, it becomes less and less likely for the MM to get executed again. The penalty accumulated between $t$ and $T$ vanishes, and the risk of reaching the inventory constraints diminishes. The MM therefore takes a bit more risk and makes a last attempt at increasing her expected P\&L, either by increasing execution probability (for null or short positions) or increasing her instantaneous profit if executed (for long positions).
\end{itemize}

\noindent Nevertheless, the following differences can be observed:

\begin{itemize}
\item Under full information (and for symmetric markets), the spreads are always symmetric in the inventory with respect the risk-neutral value, i.e., $1/b - \widetilde\delp(\cdot,n,\cdot) = \widetilde\delp(\cdot,1-n,\cdot) - 1/b$, for $0<n\leq\Nmax$. However, under partial information this symmetry is broken, as the MM increases her spreads. This holds in fact for any $0<\pi<1$ and it is suggested by the comparison of equations (\ref{partial_info_optimal_spreads}) and (\ref{full_info_optimal_spreads}), since $\theta_\pi<0$ (an increasing probability of being in the bad regime lowers the expected P\&L; see argument of \cite[Lemma 3.3]{BL0}). As a result, the full information spreads are downward biased when the exact regime is unknown. Loosely speaking, this bias is approximately proportional to the product of $\theta_\pi$ (sensitivity of the expected P\&L to observable regime changes) and $w(\pi)$ (observable order flow volatility), and inversely proportional to the observable intensity increase $\hat m(\pi)$. Intuitively, a partially informed MM faces not only the risk of the market regime shifting but also uncertainty on the current state, and must increase her spreads accordingly. This is done by considering the cost of any observable changes on the P\&L and the fluctuations in the order flow, and discounting by liquidity increases.

\item The good and bad regimes differ in the rate of order flow at the reference price (i.e., the amplitude parameters $a$ and $am$ of the orders intensities). Consequently, one could erroneously think that the partial information strategy should essentially be a one regime strategy for some intermediate parameter, such as $a\hat m(\pi)$. However, in the full information framework, an increase in the liquidity taken by clients results in the asymptotic spreads (i.e., for $t\ll T$) moving closer to the risk-neutral value (see \cite{CJ, GLFT}). On the other hand, this is not true any more in the partial information case, due to the regime risk adjustment previously mentioned that shifts the spreads. 

\item Lastly, we remark that there is a distinguished change in behaviour close to expiry, with the spreads overshooting above the risk-neutral value for some inventories, before approaching it again. Intuitively, the overshooting is due to the added regime risk adjustment, and the convergence is due to the vanishing of additional inventory risks. (Recall that the MM's main concern becomes her terminal hedging cost, should there be one, and this does not depend on the market regime.) A somewhat similar effect can be found under full information in asymmetric markets \cite{CJ}.
\end{itemize}

Figure \ref{delp_full_vs_partial_all_penalties} serves the same comparison as Figure \ref{delp_full_vs_partial_running_penalty}, but including  in this case a terminal execution cost. The observations remain mostly the same, except for a change in the terminal behaviour of the spreads, just as in the classical one regime case. Here, the spreads diverge from the risk-neutral no-costs value instead of approaching it, as the MM makes a last instant attempt to reduce her hedging cost. We remark that in both examples, a higher penalty $\zeta$ increases the difference between the full information optimal spreads of the two regimes, but it does not qualitatively change the observations made.
\begin{figure}[htb]
\hspace*{-.38cm}
	\includegraphics[scale=.33]{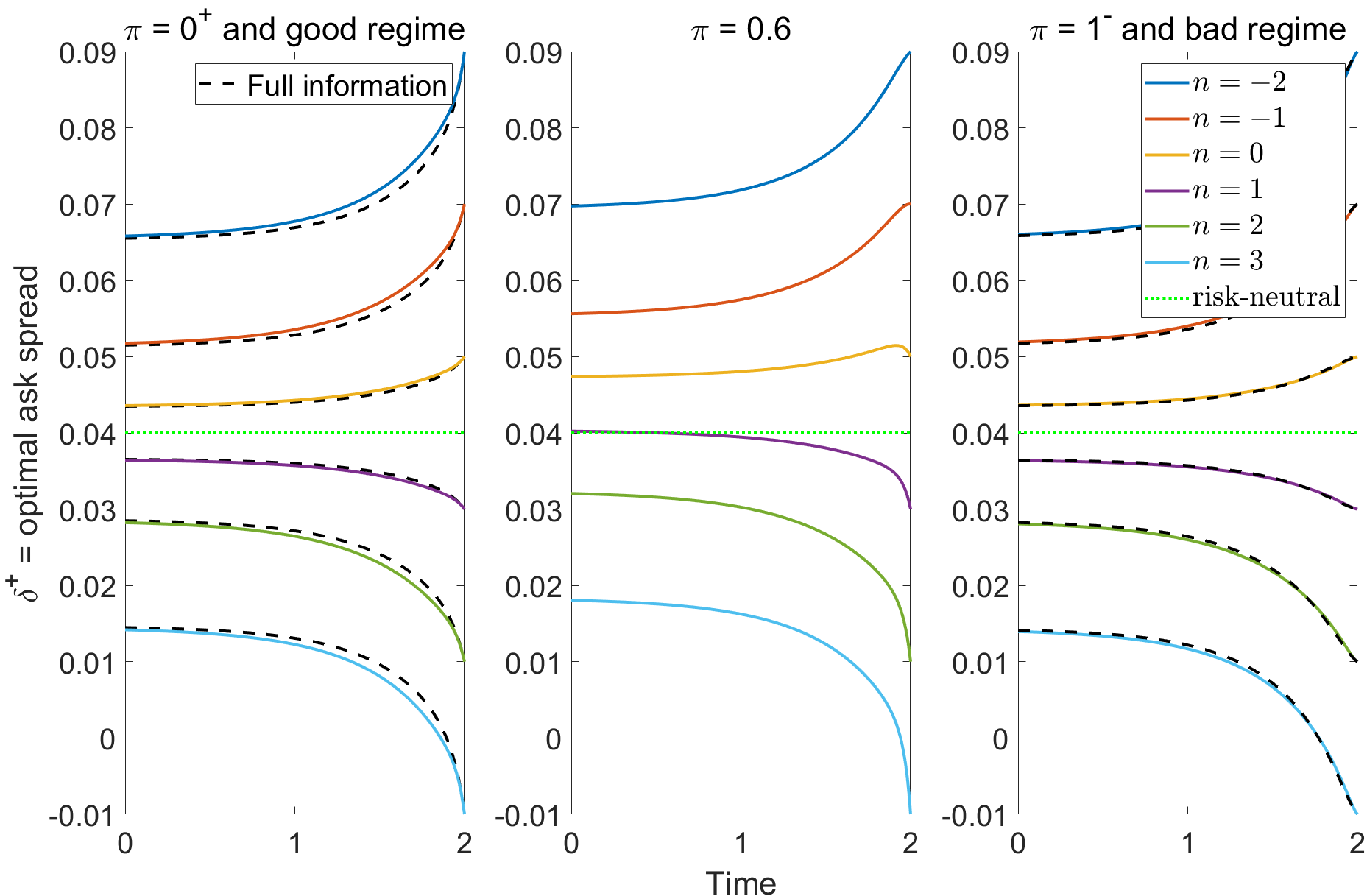}
	\caption{Optimal ask spread as a function of time, under full and partial information (dashed and solid lines resp.), and for different inventory and filter values. Inventory levels increase from top to bottom in all cases. Terminal execution cost: $\ell(n)=-cn^2$, with $c=0.01$.}
	\label{delp_full_vs_partial_all_penalties}
\end{figure}

Figure \ref{delp_partial_surf} shows the partial information optimal ask spread as a function of time and filter, for three different inventory positions: extreme short (left), flat (center) and extreme long (right), and no terminal penalty. We can see that the spread is concave on the filter, and the maximum concavity is reached for a null inventory and $\pi\approx 0.6$, gradually decreasing as the position becomes short or long. 
We had observed already in equation (\ref{partial_info_optimal_spreads}) that the partially informed MM increases her spreads to manage higher regime risk, and that this correction is approximately proportional to the expected P\&L sensitivity, $\theta_\pi$, times the observed order flow volatility $w(\pi)=(m-1)\pi(1-\pi)$. Heuristically, it seems reasonable that the change in spread should resemble the concavity of $w(\pi)$ for a risk averse MM, and that the cost of regime uncertainty is the highest for a flat position (since deviations from it increase in turn price exposure). The question of why the maximum variation is reached around $\pi\approx 0.6$ is once again differed to Section \ref{section_simulations}.

\begin{figure}[H]
\hspace*{-.1cm}
	\includegraphics[scale=.33]{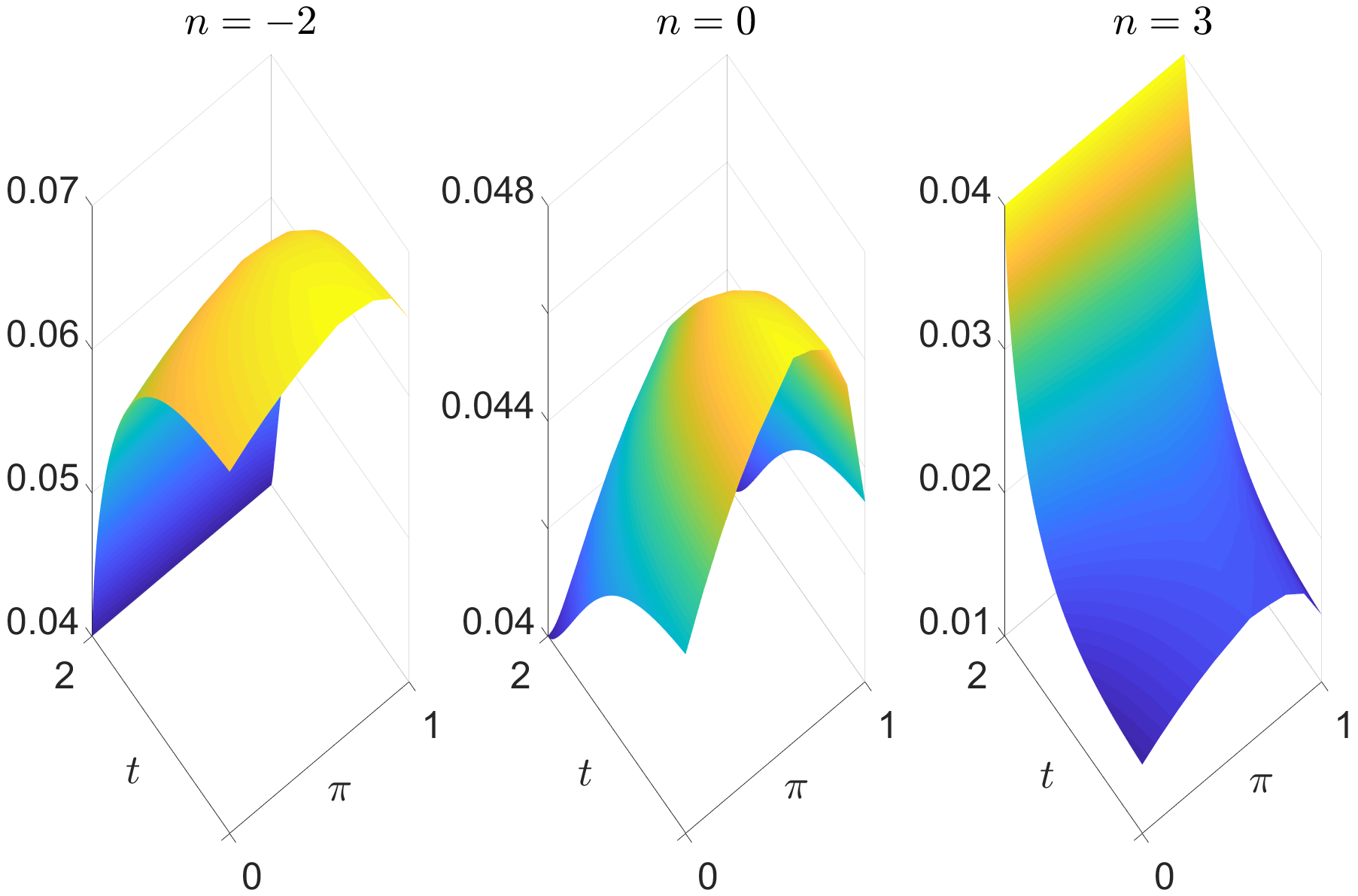}
	\caption{Optimal partial information ask spread as a function of time and filter (observable probability of bad regime), for different inventory levels. Terminal execution cost is neglected ($c=0$).}	
	\label{delp_partial_surf}
\end{figure}

\subsection{Sample paths and a closer look at the filter}
\label{section_simulations}

There are different ways to simulate point processes with stochastic intensity. A classical approach of simple implementation is the so-called \textit{thinning method} \cite{O}. This method is particularly well suited to our framework, as it can be combined with the filtering theory developed in Section \ref{s:filtering} by making use of the observable intensity of $N$ (see Theorem \ref{main_theorem_1} and \cite[Alg.3.2]{GKM}). The interested reader is referred to \cite{LTT} for optimizations in terms of the thinning bound. 

Figure \ref{simulations} shows four sample paths resulting from an optimal behaviour of the MM with incomplete information. They were obtained by jointly simulating the inventory $N$ (middle) starting from $n_0=0$, the filter $\Pi=\Pi^{\alpha,1}$ (bottom) starting from $\pi_0=0.5$, and the optimal strategy $\alpha=\big(\widehat\delm(t,N_{t^-},\Pi_{t^-}),\widehat\delp(t,N_{t^-},\Pi_{t^-})\big)$ (top), for $c=0$.

\begin{figure}[H]
\hspace*{-.7cm}
	\includegraphics[scale=.58]{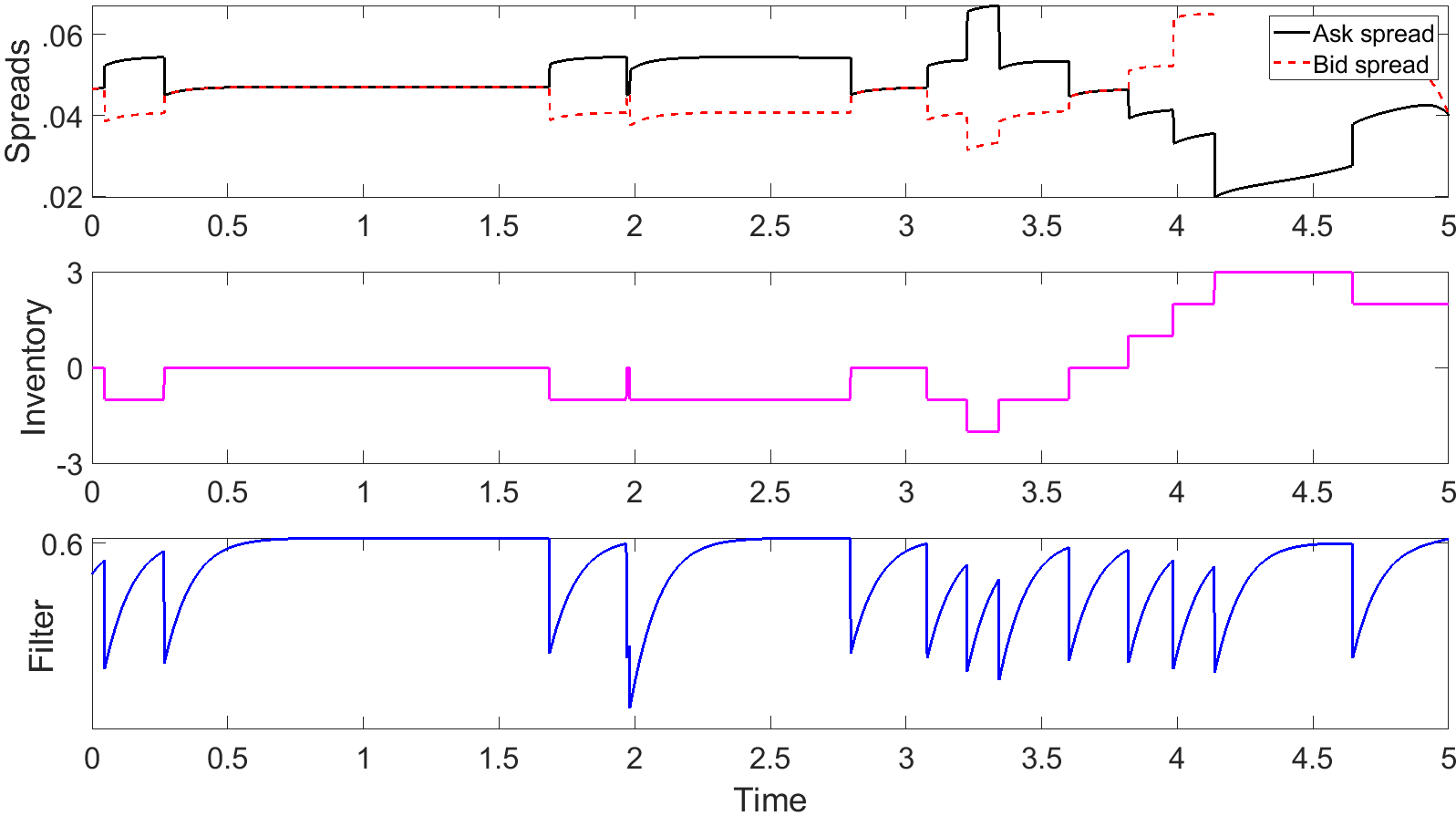}
	\caption{Sample paths of optimal partial information bid/ask spreads (first), inventory (second) and filter (third), with $c=0$, $n_0=0$ and $\pi_0=0.5$.}	
	\label{simulations}
\end{figure}

We want to analyze the behaviour of $\Pi$. Let us recall that $\Pi_t$ represents the probability of being in the bad regime (a slow market with low liquidity) given the information observed by the MM at time $t$. The MM makes her assessment based on the orders (both buy and sell) that she received so far, which is the same as looking at the evolution of her inventory. If we take a close look at what happens between any two consecutive orders, we see that there is a pattern that repeats itself. As time passes without the MM receiving any order, she deems more likely that the market is in the bad regime, and so $\Pi_t$ increases. If enough time goes by in this way, then $\Pi_t$ gets asymptotically closer to some preponderant ``equilibrium" value which is around 0.6. (Recall that this is the same value we have encountered before as having the greatest effect on the spreads.) But as soon as an order arrives, the MM revisits her probabilities, as being in the good regime of high liquidity seems a lot more likely. This is why we see that $\Pi_t$ jumps downwards with each trade (i.e., at the jump times of the inventory).

So what makes $\pi\approx 0.6$ so special? If we look at the SDE (\ref{Pi_equivalent}) governing the filter dynamics, it now reads:

\begin{equation}
d\Pi_t = \left(\hat q(\Pi_t) + w(\Pi_t)a\big(\mathbbm 1^-_t e^{-b\widehat{\delta_t^-}}+\mathbbm 1^+_t e^{-b\widehat{\delta_t^+}}\big)\right)dt + \Pi_{t^-}\left(\frac{1}{\hat m\big(\Pi_{t^-}\big)}-1\right)\Big(d\Nm_t + d\Np_t \Big).
\end{equation}

In between jumps, each path of $\Pi_t$ evolves according to an ODE. If enough time passes by without another order arriving, we can consider the asymptotic behaviour for $t\to T$. For $c=0$, this leads to
$$
\frac{d\Pi_t}{dt}\approx \hat q(\Pi_t) + w(\Pi_t)\frac{a}{e}\beta,\quad\mbox{with }\beta=1\mbox{ or }\beta=2,
$$
depending on the current inventory level. The right-hand side of this ODE is a concave parabola in $\Pi_t$ with a negative root and another one, $\pi^*$, between 0 and 1. For the parameter values in Table \ref{table_parameters}, it is either $\pi^*\approx 0.572$ for $\beta=1$, or $\pi^*\approx 0.636$ for $\beta=2$. Therefore, $\pi^*$ is an attractor (i.e., a stable equilibrium) for the filter, which explains the tendency of $\Pi_t$ towards 0.6 and the central role of this value. Indeed, as the filter approaches one of its attractors, the MM exhausts all the information she has and can make no further assessment on the state of the market until she is hit by another order. In a sense, this is when she faces the most uncertainty, and she increases her spreads accordingly to manage regime risk. 

\appendix
\section{Additional proofs}
\label{appendix}
 \setcounter{subsection}{1}

\begin{proof}[\textbf{Proof of Proposition \emph{\ref{Girsanov1}}}]
\phantomsection
\label{proof:Girsanov1}
	
$\Nm,\Np$ are finite variation (FV) processes with no common jumps, hence $[\Nm,\Np]=0$ $\mathbb Q$-a.s. Therefore, the multiplicativity of the stochastic exponential \cite[p.138]{JS} and the exponential formula for FV processes \cite[p.337 T4]{B} yield the explicit expression $Z^\alpha= Z^{\alpha,-} Z^{\alpha,+}$, with
	\begin{equation}
	\label{exponential_explicit}
	Z^{\alpha,\pm}_t=\exp\left(\int_0^t \big(1-\Lpm_{Y_{u^-}}(\delp_u)\big)\mathbbm 1^{\pm}_u du\right)\prod_{\substack{u\leq t:\\ \Delta\Npm_u\neq 0}}\Lpm_{Y_{u^-}}(\delp_u).
	\end{equation}
	Then \textit{(i)} follows from the strict positivity and boundedness of $\Lpm_i(\delpm)$ for $1\leq i\leq k$, as in \cite[p.168 T4]{B} with terminal time $T$. The only difference is that in our case the $\mathbb Q$-intensity of $\Npm$ is $\mathbbm 1_{\{\mp N_{t^-}<\Nmax\}}\leq 1$ instead of $1$. The proof remains the same though, simply recalling that the $\mathbb Q$-moment generating function of $\Npm_T$ is dominated by the one of a standard Poisson random variable (see (\ref{bounded_moments})). Uniform integrability is immediate.
	
\textit{(i)} guarantees that $\mPa$, as defined in \textit{(ii)}, is an equivalent probability measure. The shape of the $\mPa$-intensities of $\Nm,\Np$ is due to \cite[p.166 T3]{B}. 
	
The fact that $W$ is still a Wiener process under $\mPa$ is a consequence of the Girsanov--Meyer Theorem \cite[p.132 Thm.35]{P} and Levy's characterization Theorem.

As for $Y$, note first that its initial distribution does not change as $Z^\alpha_0=1$, and the same goes for its infinitesimal generator operator. To see this, consider the $\mathbb Q$-generator operator $\mathcal A^Y_t:\mathbb R^E\to\mathbb R^E$, $\mathcal A^Y_tf(i)= \sum_{j=1}^k q_t^{ij}f(j)$. We know the process $M_t=f(Y_t)-f(Y_0) - \int_0^t\mathcal A^Y_uf(Y_u)du$ is a $\mathbb Q$-local martingale. Once again, by the Girsanov--Meyer Theorem., $M$ is also a $\mPa$-local martingale ($[M,Z]=0$ since $Y,\Nm,\Np$ are FV processes satisfying (\ref{no_common_jumps})). Moreover, being bounded, $M$ is a true  $\mPa$-martingale and $Y$ solves, for $(\mathbb F,\mPa)$, a well-posed martingale problem for $(\mathcal A^Y_t, \mu_0)$. This implies $Y$ is a $\mPa$-Markov chain with a uniquely determined law \cite[p.184 Thm.4.2]{EK}.\footnote{Although our martingale problem is non-homogeneous in time, $(Q_t)$ is deterministic, so this does not represent a problem.} That $Q$ is also the $\mPa$-generator matrix follows from uniqueness.    
	
(\ref{no_common_jumps}) clearly remains unchanged under an equivalent change of probability measure.
\end{proof}	
	
\begin{proof}[\textbf{Proof of Proposition \emph{\ref{Girsanov2}}}]
\phantomsection
\label{proof:Girsanov2}

\textit{(i)} and \textit{(ii)} are proved just as in Proposition \ref{Girsanov1} (the processes $1/\Lpm_i(\delpm)$ are strictly positive and bounded for $1\leq i\leq k$). The $\mQa$-independence of $Y,N,W$ is a consequence of \cite[p.543 Lem.9.5.4.1]{JYC}. That is, the $(\mQa,\mathbb F)$-martingales $Y - \int_0^\cdot\sum_{j=1}^k q_u^{Y_{u^-},j}j du,\ N-\int_0^\cdot(\mathbbm 1^-_u - \mathbbm 1^+_u)du\mbox{ and }W$ have the predictable representation property with respect to $\mathbb F^Y,\ \mathbb F^N\mbox{ and }\mathbb F^W$ resp. (see \cite[p.239 T8]{B} for compensated counting measures) and they are $(\mathbb F,\mQa)$-orthogonal, which implies their independence.
	
The explicit expressions for the stochastic exponentials (see (\ref{exponential_explicit})) show straightforwardly that $Z^\alpha\bar{Z}^\alpha=1$, proving \textit{(iii)}.
\textit{(iv)} is due to \cite[p.64 T8]{B}.

\end{proof}

\begin{proof}[\textbf{Proof of Proposition \emph{\ref{observable_intensities}}}]
\phantomsection
\label{proof:observable_intensities}

We prove it just for $\Np$, the others being analogous, and we omit $\alpha$ for simplicity. It is clear that $\widehat{\lp}$ is predictable. We need to check that for any $\mathbb F^{W,N}$-predictable process $\psi\geq 0$, 
	$$
	\mathbb E\left[\int_0^T\psi_t\widehat{\lp_t}dt\right] = \mathbb E\left[\int_0^T\psi_td\Np_t\right].
	$$
For any $1\leq i\leq k$, each path of the càdlàg processes $\pii$ and $\mathbbm 1_{\{Y_\cdot=i\}}$ has only countably many jumps. We can therefore interchange these processes and their left limits when integrating with respect to $dt$. By properties of the conditional expectation and Fubini's Theorem,
\begin{align*}
	&\mathbb E\left[\int_0^T\psi_t\widehat{\lp_t}dt\right]=\mathbb E\left[\int_0^T \psi_t\mathbbm 1^+_t\sum_{i=1}^k\pii_{t^-}\Lip(\delp_t) dt\right] = \mathbb E\left[\int_0^T \psi_t\mathbbm 1^+_t\sum_{i=1}^k\pii_t\Lip(\delp_t) dt\right]\\
																																	& = \mathbb E\left[\int_0^T \mathbb \psi_t\mathbbm 1^+_t\sum_{i=1}^k\mathbb E\big[\mathbbm 1_{\{Y_t=i\}}\big|\mathcal F^{N,W}_t\big]\Lip(\delp_t)dt\right] =\int_0^T \mathbb E\left[\mathbb E\Big[\psi_t\mathbbm 1^+_t\sum_{i=1}^k\mathbbm 1_{\{Y_t=i\}}\Lip(\delp_t) \big|\mathcal F^{N,W}_t\Big]\right]dt\\
																																	& = \int_0^T \mathbb E\Big[\psi_t\mathbbm 1^+_t\sum_{i=1}^k\mathbbm 1_{\{Y_t=i\}}\Lip(\delp_t)  \Big]dt= \mathbb E\left[\int_0^T \psi_t\mathbbm 1^+_t\sum_{i=1}^k\mathbbm 1_{\{Y_{t^-}=i\}}\Lip(\delp_t) dt\right] =\mathbb E\left[\int_0^T \psi_t\lp_t dt\right]\\
																																	&= \mathbb E\left[\int_0^T \psi_td\Np_t\right].
\end{align*}
\end{proof}

\begin{proof}[\textbf{Proof of Proposition \emph{\ref{proposition_Pi}}}]
\phantomsection
\label{proof:proposition_Pi}

Let us check that $\pia$ solves (\ref{Pi_equivalent}). Clearly, the constraint and the initial condition are satisfied. The verification of the SDEs is due to \cite[Prop.3.3]{CEFS} (with more details in \cite[App.A, Lemma A.2 and Prop.3.3]{CEFSv3}), albeit some considerations need to be made. 
	
	On the one hand, the authors work with a pure jump model, with strategies adapted to the natural filtration of the driving jump process only (i.e., there is no diffusion) and constant generator matrix for the Markov chain. However, \textit{mutatis mutandis} the former differences yield no major change in the proofs. 
	
	On the other hand, the main assumption of the authors, \cite[Asm.2.1]{CEFS}, postulates the existence of some deterministic measure $\tilde{\eta}^{N}(dz)$ on $\mathbb R$ with compact support\footnote{In \cite{CEFS} the support is assumed to be a subset of $(-1,\infty)$, but this is only for a ``return (or yield) process" as in their case.} such that for all $ i\in E,\ \delp,\delm\in \overline{(\delmin,\delmax)},\ n\in [-\Nmax,\Nmax]\cap\mathbb Z$, the measure $\eta^N(i, \delp,\delm,n,dz)$ is equivalent to $\tilde{\eta}^{N}(dz)$ and the Radon-Nikodym derivative $d\eta^{N}(i, \delp,\delm,n,\cdot)/d\tilde{\eta}^N$ is uniformly bounded and bounded away from zero $\tilde{\eta}^{N}(dz)$-a.s. 
	
	Since the spread processes are fixed and bounded, we can assume without loss of generality that $\delmin$ is finite. Setting $\tilde{\eta}^N(dz) \defeq m_{1}(dz)  + m_{-1}(dz)$ we see straightforwardly that the two measures are equivalent with derivative 
	\begin{equation}
	\frac{d\eta^{N}(i, \delm,\delp,n,\cdot)}{d\tilde{\eta}^N}(z) = \Lim(\delm) \mathbbm 1_{\{n<\Nmax\}} \mathbbm 1_{\{z= 1\}}+ \Lip(\delp)\mathbbm 1_{\{-n<\Nmax\}} \mathbbm 1_{\{z=-1\}},
	\end{equation} 
	uniformly bounded by $\Lim(\delmin) + \Lip(\delmin)$. However, our model allows for $d\eta^{N}/d\tilde{\eta}^N(z)=0$, which is a consequence of having vanishing intensities $\lpm$. This poses no issue nonetheless, as $d\eta^{N}/d\tilde{\eta}^N(z)>0$ is only used in \cite[Prop.3.3]{CEFS} to guarantee $\bar{Z}^\alpha>0$ (see Proposition \ref{Girsanov2}). This condition, also satisfied in our model,\footnote{This was ultimately a consequence of the decomposition in Assumptions \ref{assumptions_intensities} \textit{(\ref{assumptions_intensities_lambda})}, that allowed for the vanishing factors of the intensities to be secluded as the reference probability intensities.} allows to go from the physical probability $\mPa$ to a reference probability $\mQa$ and backwards.
	
	We turn now to the proof of uniqueness. We remark first that the jump height coefficients in (\ref{Pi_equivalent}) will typically not be Lipschitz (classical results for SDEs such as \cite[p.253 Thm.7]{P} cannot be applied) and the paths of the spreads need not be continuous between the jump times of $N$ (ruling out the most classical results of ODEs \cite{CL}). Nevertheless, we can still follow a pathwise ODEs approach. 
	
	Let us fix a path and verify uniqueness inductively on the intervals $[\tau_m,\tau_{m+1})$, where $\tau_0\defeq 0$ and $0<\tau_1<\tau_2<\dots<\tau_M=T$ are the jump times of $N$ (including the terminal time $T$ even if there is no jump at that point).
	Set $A^{ji}_t=\mathbbm 1_t^-(\Ljm-\Lim)(\delm_t) + \mathbbm 1_t^+(\Ljp- \Lip)(\delp_t)$. Then $A^{ji}$ is bounded for all $1\leq i,j\leq k$. Now observe that any càdlàg process $\tilde\Pi$, solving the constrained system of SDEs (\ref{Pi_equivalent}), must solve pathwise for $m=0,1,\dots,M-1$ the following system of ODEs in integral form, for $t\in [\tau_m,\tau_{m+1})$:
	\begin{equation}
	\label{Pi_ODE}
	\tilde\Pi^i_t = R^i_m\big(\tilde\Pi_{\tau_m^-}\big) + \int_{\tau_m}^t\sum_{j=1}^k \left(q_u^{ji}\tilde\Pi^j_u+\tilde\Pi^i_u\tilde\Pi^j_u A^{ji}_u \right)du,
	\end{equation}
	with $R^i_m\big(\tilde\Pi_{\tau_m^-}\big)\defeq\tilde\Pi^i_{\tau_m^-}\Lipm(\delpm_{\tau_m})/\sum_{j=1}^k\tilde\Pi^j_{\tau_m^-}\Ljpm(\delpm_{\tau_m})$ if $\Delta N_{\tau_m}=\mp 1,\ m>0$, and $R^i_0\big(\tilde\Pi_{\tau_0^-}\big) = \mu^i_0$. 
	Using that $\tilde\Pi\in\Delta$ is bounded, and elementary algebra of bounded Lipschitz functions, it follows that $f^i:[\tau_m,\tau_{m+1})\times\Delta\to\mathbb R$, defined by $f^i(u,\pi)= \sum_{j=1}^k (q_u^{ji}\pi^j+\pi^i\pi^j A^{ji}_u)$ is Lipschitz in $\pi$, uniformly in $u$. Let $K$ be the maximum Lipschitz constant of $f^l$ for $1\leq l\leq k$ and suppose $\pia_{\tau_m^-}=\tilde\Pi_{\tau_m^-}$ (clearly satisfied for $m=0$). Then (\ref{Pi_ODE}) yields $\|\pia_t-\tilde\Pi_t\|\leq K\int_{\tau_m}^t\|\pia_u-\tilde\Pi_u\|du$, implying $\pia_t=\tilde\Pi_t$ on $[\tau_m,\tau_{m+1})$ by Grönwall's inequality. As a consequence, the equality on $[0,T)$ follows by induction. It must clearly hold at time $T$ as well, either by continuity or (if there is a jump) because $\pia_T=R_M\big(\pia_{\tau_M^-}\big)=R_M\big(\tilde\Pi_{\tau_M^-}\big)= \tilde\Pi_T$. 
	\end{proof}

\begin{proof}[\textbf{Proof of Proposition \emph{\ref{open_simplex}}}]
\phantomsection
\label{proof:open_simplex}

We want to show that $\piia_t>0$ for all $1\leq i\leq k,\ 0\leq t\leq T$ a.s. We proceed by induction on the jump times of $N$, for each path, as at the end of the proof of Proposition \ref{proposition_Pi}. Using the same notations, let $1\leq i \leq k$ and suppose $\pija_{\tau_m^-}>0$ for all $1\leq j\leq k$ (satisfied for $m=0$ by assumption). Then (\ref{Pi_ODE}), Assumptions \ref{assumptions_Q} and the fact that $A^{ii}\equiv 0$ by definition, show that $\piia$ is absolutely continuous on the interval $[\tau_m,\tau_{m+1})$ and satisfies 
	\begin{equation}
	\label{aux1}
	\begin{split}
	(\piia_t)' & = \sum_{j=1}^k \left(q_t^{ji}\pija_t+\piia_t\pija_t A^{ji}_t \right) = \piia_t\Big(q^{ii}_t + \sum_{j\neq i}\pija_t A^{ji}_t \Big) + \sum_{j\neq i}q^{ji}_t\pija_t\\
						 & \geq\piia_t\Big(q^{ii}_t + \sum_{j\neq i}\pija_t A^{ji}_t\Big),
	\end{split}
	\end{equation}
	for $dt$-a.e. $t\in [\tau_m,\tau_{m+1})$, subject to $\piia_{\tau_m}=R^i_m(\pia_{\tau_m^-})>0$. Let us set $s\defeq \sup ([\tau_m,\tau_{m+1}) \cap \{t\in [0,T] : \piia_t >0\})$. We need to prove that $s=\tau_{m+1}$. By the continuity of $\piia$, it must be $s>\tau_m$. Consequently, (\ref{aux1}) and the absolute continuity of $\log\piia$ on $[\tau_m,t]\subset[\tau_m,s)$ yield 
	$$
	\piia_t\geq R^i_m(\pia_{\tau_m^-})\exp\Big(\int_{\tau_m}^t\big(q^{ii}_u + \sum_{j\neq i}\pija_u A^{ji}_u \big)du\Big)\quad\mbox{ for }dt\mbox{-a.e. }t\in [\tau_m,s).
	$$ 
	If it were $s<\tau_{m+1}$, the continuity of $\piia$ again and the former inequality would imply $0=\piia_s>0$, proving by contradiction that $s=\tau_{m+1}$ and $\piia$ is positive on the whole interval $[\tau_m,\tau_{m+1})$. Positivity on $[0,T)$ now follows by induction, and it must clearly hold at time $T$ as well, as either there is a jump or we can reason as we just did.
\end{proof}

\begin{proof}[\textbf{Proof of Theorem \emph{\ref{main_theorem_2}}}]
\phantomsection
\label{proof:main_theorem2}
\textit{Case $\Nmax<+\infty$ and $\gamma>0$}: Let us write $\Psi=\frac{1 + \Upsilon}{\gamma}$, with 
\begin{equation*}
\Upsilon(t,n,\pi) = \sup_{\alpha\in\tilde{\mathcal A}_t^2}\tilde{\mathbb E}^{\alpha,t,n,\pi}\left[ 
\exp\Big(-\gamma\big(\tilde P_{t,T}^{\alpha,n,\pi}+\ell(\Nt_T)\big)\Big) \Big(-\exp\big(\gamma\ell(\Nt_T)\big) \Big)\right].
\end{equation*}
Then $\Upsilon$ can be regarded as the value function of an optimization problem in the standard Bolza-Lagrange formulation, i.e., 
$$
\Upsilon(a)= \sup_{\alpha\in\tilde{\mathcal A}_t^2}\tilde{\mathbb E}^{\alpha,a}\left[\int_t^T D^{\alpha,a}_{t,u}f(u,A^{\alpha,a}_u,\alpha_u) du + D^{\alpha,a}_{t,T}g(A^{\alpha,a}_T)\right],
$$
where the state variable $A^{\alpha,a}_u=(u,\Nt_u,\pit_u)$ is a PDMP with bounded state space $[0,T]\times(\mathbb Z\cap[-\Nmax,\Nmax])\times\Delta^\circ$ and initial condition $a=(t,n,\pi)$, $D^{\alpha,a}_{t,u}=\exp\Big(-\int_t^v\rho(u,A^{\alpha,a}_u,\alpha_u)du\Big)$ is the discount factor and $f,g,\rho$ are bounded functions. (These functions are bounded thanks to the control space being bounded.) The continuity of $\Upsilon$ can be proved now in the same way as in \cite[Thm.4.10]{CEFS} albeit in a more straightforward manner. This is due to the boundedness of $f,g,\rho$, the fact that $\pit$ never visits the relative border of the simplex $\Delta$, and that there is no exit time of the state space other than the terminal time. Assumptions \cite[Asm.4.7]{CEFS} are clearly verified in our model, and \cite[Asm.2.1]{CEFS} has already been accounted for at the beginning of the proof of Proposition \ref{proposition_Pi}.\footnote{An additional detail now is that \cite[Asm.2.1]{CEFS} is also used in \cite[Lem.4.1]{CEFS}, which states that the drift coefficient in equation (\ref{Pi_equivalent}) is Lipschitz in the state variable, uniform in time and control. This is routinely verified in our case, under our new assumption: $-\infty<\delmin<\delmax<\infty$.} We remark that in our case the bounding function (see \cite[Lem.4.6]{CEFS}) can be taken simply as $b(t,n,\pi)=\exp(\eta(T-t))$, for an $\eta>0$ large enough to prove contractiveness. 

Having proved the continuity, the same proof of \cite[Thm.5.3]{CEFS} (or \cite[Thm.7.5]{DF}) shows that $\Upsilon$ is the unique continuous viscosity solution of its standard HJB equation. We remark once again that our case is simpler, in that $\Upsilon$ is bounded and we do not have any boundary conditions other than the terminal time condition. In particular, there is no need for additional assumptions on
the growth of $\Upsilon$. 

Finally, the result for $\Theta$ is obtained via the two increasing diffeomorphic transformations $\Psi=\frac{1 + \Upsilon}{\gamma}$ and $\Theta=U_\gamma^{-1}\circ\Psi$.

\textit{Case $\Nmax<+\infty$ and $\gamma=0$}: The only difference with the previous case is that the Bolza-Lagrange representation of the problem is obtained directly, since 
$$
\Theta(t,n,\pi)=\sup_{\alpha\in\tilde{\mathcal A}^2_t}\mathbb E^{\alpha,t,n,\pi}\left[\int_t^T\big\{\delm_u \hlmtu + \delp_u \hlptu + \mu \Nt_u -\frac{1}{2}\sigma^2\zeta(\Nt_u)^2\big\}du - \ell(\Nt_T) \right],
$$
with no need for any transformation.

\textit{Case $\Nmax=+\infty$ and $\gamma=\zeta=0\equiv\ell$}: 
The same as the latter case but working instead with the state variable $(u,\pit)$ and the value function $\Phi$. 
\end{proof}

\begin{proof}[\textbf{Proof of Proposition \emph{\ref{comparison_ODEs}}}]
\phantomsection
\label{proof:comparison_ODEs}

Suppose first $I=[\tau,T]$ for some $0\leq\tau<T$ and let $\ve>0$. Since $\Nmax<\infty$, there exists $(t_\ve,n_\ve,i_\ve)\in[\tau,T]\times\mathcal I$ such that
\begin{equation}
\label{min}
\overline\theta(t_\ve,n_\ve,i_\ve)-\underline\theta(t_\ve,n_\ve,i_\ve)+\ve(T-t_\ve) = \min_
{(t,n,i)\in[\tau,T]\times\mathcal I}\overline\theta(t,n,i)-\underline\theta(t,n,i)+\ve(T-t).
\end{equation}
If $t_\ve<T$, then we must have
$$\overline\theta_t(t_\ve,n_\ve,i_\ve)-\underline\theta_t(t_\ve,n_\ve,i_\ve)\geq \ve.$$
Let us see that the left-hand side is non-positive. By (\ref{HJB_Theta_full_info_supersol}) and (\ref{HJB_Theta_full_info_subsol}),
\begin{equation*}
\begin{split}
& \overline\theta_t(t_\ve,n_\ve,i_\ve)-\underline\theta_t(t_\ve,n_\ve,i_\ve) \leq\sum_{j\neq i_\ve} q_t^{i_\ve j}\Big(U_\gamma\big(\underline\theta(t_\ve,n_\ve,j)-\underline\theta(t_\ve,n_\ve,i_\ve)\big)-U_\gamma\big(\overline\theta(t_\ve,n_\ve,j)-\overline\theta(t_\ve,n_\ve,i_\ve)\big)\Big)\\
& + \mathbbm 1_{\{n<\Nmax\}}\Big(\Hm_i\big(\underline\theta(t_\ve,n_\ve+1,i_\ve)-\underline\theta(t_\ve,n_\ve,i_\ve)\big)-\Hm_i\big(\overline\theta(t_\ve,n_\ve+1,i_\ve)-\overline\theta(t_\ve,n_\ve,i_\ve)\big)\Big)\\
& + \mathbbm 1_{\{-n<\Nmax\}}\Big(\Hp_i\big(\underline\theta(t_\ve,n_\ve-1,i_\ve)-\underline\theta(t_\ve,n_\ve,i_\ve)\big)-\Hp_i\big(\overline\theta(t_\ve,n_\ve-1,i_\ve)-\overline\theta(t_\ve,n_\ve,i_\ve)\big).
\end{split}
\end{equation*} 
$\Hpm_i$ increasing (resp. $U_\gamma$ increasing) and (\ref{min}) imply that the last two terms (resp. the first one) are non-positive. We must have then that $t_\ve=T$, and due to (\ref{min}) and (\ref{terminal_cond_comparison}) for all $(t,n,i)\in [0,T]\times\mathcal I$: 
\begin{align*}
\overline\theta(t,n,i)-\underline\theta(t,n,i)+\ve(T-t)&\geq\overline\theta(T,n_\ve,i_\ve)-\underline\theta(T,n_\ve,i_\ve)+\ve(T-T)\geq 0\\
																			\overline\theta(t,n,i)&\geq	\underline\theta(t,n,i)-\ve(T-t).
\end{align*}
Since $\ve>0$ was arbitrary, we obtain the desired result. 

The case $I=(\tau,T]$ is now a consequence of comparing $\underline\theta$ and $\overline\theta$ on intervals of the form $[t_n,T]\subseteq (\tau,T]$ with $t_n\searrow\tau$.
\end{proof}

\bibliographystyle{amsalpha}
\bibliography{references}
\addcontentsline{toc}{section}{References}

\end{document}